\documentclass[11pt]{article}
\usepackage[margin=.8in,left=.8in]{geometry}
\usepackage{amsmath}
\usepackage{amsfonts}
\usepackage{amssymb}
\usepackage{amsthm}
\usepackage{mathtools}
\usepackage{subcaption}
\usepackage{graphicx}
\usepackage{color}
\usepackage{xcolor}
\usepackage{xfrac}
\usepackage{stmaryrd}
\usepackage[utf8]{inputenc}
\usepackage{rotating}
\usepackage{hyperref}
\usepackage[all]{xy}

\usepackage{lmodern} % for  \bf and \it at the same time
\usepackage{ulem}

\usepackage{tikz}
\usetikzlibrary{decorations.pathmorphing}
\usepackage{tikz-cd}
\usetikzlibrary{calc}
\usetikzlibrary{matrix}
\usetikzlibrary{positioning,arrows,shapes}
\usetikzlibrary{intersections,shapes.arrows, calc}
\usetikzlibrary{patterns}
\usetikzlibrary{mindmap} % LATEX and plain TEX
\usetikzlibrary[mindmap] % ConTEXt

\usepackage{lipsum}
\usepackage{adjustbox}

\definecolor{plum}{rgb}{0.36078, 0.20784, 0.4}
\definecolor{chameleon}{rgb}{0.30588, 0.60392, 0.023529}
\definecolor{cornflower}{rgb}{0.12549, 0.29020, 0.52941}
\definecolor{scarlet}{rgb}{0.8, 0, 0}
\definecolor{brick}{rgb}{0.64314, 0, 0}
\definecolor{sunrise}{rgb}{0.80784, 0.36078, 0}
\definecolor{lightblue}{rgb}{0.15,0.35,0.75}
\definecolor{carolina}{RGB}{153, 186, 221}
\definecolor{darkblue}{rgb}{0.05,0.25,0.65}

\usepackage{multirow}

\usepackage{stmaryrd}    % for \sslash

\usepackage{enumerate} %Roman enumerate

\usepackage{helvet}   %vertical spacing in tabular

\usepackage{mathptmx}
\usepackage{amsmath}
\usepackage{graphicx}

\usepackage{floatflt}  %Floating tables
\usepackage{array}     % specifying size of tables
\newcolumntype{L}[1]{>{\raggedright\let\newline\\\arraybackslash\hspace{0pt}}m{#1}}
\newcolumntype{C}[1]{>{\centering\let\newline\\\arraybackslash\hspace{0pt}}m{#1}}
\newcolumntype{R}[1]{>{\raggedleft\let\newline\\\arraybackslash\hspace{0pt}}m{#1}}

\makeatletter
\newcommand{\raisemath}[1]{\mathpalette{\raisem@th{#1}}}
\newcommand{\raisem@th}[3]{\raisebox{#1}{$#2#3$}}
\makeatother

\usepackage{tabularx}   %For sizing tabular entries

\usepackage[new]{old-arrows}   %For \longhookrightarrow

\newdir{> }{{}*!/10pt/@{>}}

\makeatletter
\newif\if@sup
\newtoks\@sups
\def\append@sup#1{\edef\act{\noexpand\@sups={\the\@sups #1}}\act}%
\def\reset@sup{\@supfalse\@sups={}}%
\def\mk@scripts#1#2{\if #2/ \if@sup ^{\the\@sups}\fi \else%
  \ifx #1_ \if@sup ^{\the\@sups}\reset@sup \fi {}_{#2}%
  \else \append@sup#2 \@suptrue \fi%
  \expandafter\mk@scripts\fi}
\def\tensor#1#2{\reset@sup#1\mk@scripts#2_/}
\def\multiscripts#1#2#3{\reset@sup{}\mk@scripts#1_/#2%
  \reset@sup\mk@scripts#3_/}
\makeatother

\makeatletter
\newbox\slashbox \setbox\slashbox=\hbox{$/$}
\def\itex@pslash#1{\setbox\@tempboxa=\hbox{$#1$}
  \@tempdima=0.5\wd\slashbox \advance\@tempdima 0.5\wd\@tempboxa
  \copy\slashbox \kern-\@tempdima \box\@tempboxa}
\def\slash{\protect\itex@pslash}
\makeatother

\def\clap#1{\hbox to 0pt{\hss#1\hss}}

\let\oldroot\root
\def\root#1#2{\oldroot #1 \of{#2}}
\renewcommand{\sqrt}[2][]{\oldroot #1 \of{#2}}

\DeclareSymbolFont{symbolsC}{U}{txsyc}{m}{n}
\SetSymbolFont{symbolsC}{bold}{U}{txsyc}{bx}{n}
\DeclareFontSubstitution{U}{txsyc}{m}{n}

\DeclareSymbolFont{stmry}{U}{stmry}{m}{n}
\SetSymbolFont{stmry}{bold}{U}{stmry}{b}{n}

\DeclareFontFamily{OMX}{MnSymbolE}{}
\DeclareSymbolFont{mnomx}{OMX}{MnSymbolE}{m}{n}
\SetSymbolFont{mnomx}{bold}{OMX}{MnSymbolE}{b}{n}
\DeclareFontShape{OMX}{MnSymbolE}{m}{n}{
    <-6>  MnSymbolE5
   <6-7>  MnSymbolE6
   <7-8>  MnSymbolE7
   <8-9>  MnSymbolE8
   <9-10> MnSymbolE9
  <10-12> MnSymbolE10
  <12->   MnSymbolE12}{}

\makeatletter
\def\Decl@Mn@Delim#1#2#3#4{%
  \if\relax\noexpand#1%
    \let#1\undefined
  \fi
  \DeclareMathDelimiter{#1}{#2}{#3}{#4}{#3}{#4}}
\def\Decl@Mn@Open#1#2#3{\Decl@Mn@Delim{#1}{\mathopen}{#2}{#3}}
\def\Decl@Mn@Close#1#2#3{\Decl@Mn@Delim{#1}{\mathclose}{#2}{#3}}
\Decl@Mn@Open{\llangle}{mnomx}{'164}
\Decl@Mn@Close{\rrangle}{mnomx}{'171}
\Decl@Mn@Open{\lmoustache}{mnomx}{'245}
\Decl@Mn@Close{\rmoustache}{mnomx}{'244}
\makeatother

\makeatletter
\DeclareRobustCommand\widecheck[1]{{\mathpalette\@widecheck{#1}}}
\def\@widecheck#1#2{%
    \setbox\z@\hbox{\m@th$#1#2$}%
    \setbox\tw@\hbox{\m@th$#1%
       \widehat{%
          \vrule\@width\z@\@height\ht\z@
          \vrule\@height\z@\@width\wd\z@}$}%
    \dp\tw@-\ht\z@
    \@tempdima\ht\z@ \advance\@tempdima2\ht\tw@ \divide\@tempdima\thr@@
    \setbox\tw@\hbox{%
       \raise\@tempdima\hbox{\scalebox{1}[-1]{\lower\@tempdima\box
\tw@}}}%
    {\ooalign{\box\tw@ \cr \box\z@}}}
\makeatother

\makeatletter
\def\udots{\mathinner{\mkern2mu\raise\p@\hbox{.}
\mkern2mu\raise4\p@\hbox{.}\mkern1mu
\raise7\p@\vbox{\kern7\p@\hbox{.}}\mkern1mu}}
\makeatother

\newcommand{\sset}{s\mathcal{S}{\mathrm{et}}}
\newcommand{\set}{{\mathcal{S}\mathrm{et}}}

\newcommand{\PSh}{{\mathcal{P}\mathrm{Sh}}}
\def\sPSh{\PSh_\Delta}

\newcommand{\cartsp}{{\sf Cart}}

\def\hocolim{\mathop{\rm hocolim}}

\def\1{{\bf 1}}

\def\<{\langle}
\def\>{\rangle}

\newcommand{\C}{\ensuremath{\mathbb{C}}}
\newcommand{\HH}{\ensuremath{\mathbb{H}}}
\newcommand{\into}{\hookrightarrow}

\newcommand{\R}{\ensuremath{\mathbb R}}
\newcommand{\Z}{\ensuremath{\mathbb Z}}
\newcommand{\Q}{\ensuremath{\mathbb Q}}

\renewcommand{\(}{\begin{equation}}
\renewcommand{\)}{\end{equation}}
\newcommand{\bea}{\begin{eqnarray*}}
\newcommand{\eea}{\end{eqnarray*}}

\usepackage{cleveref}

\crefformat{section}{\S#2#1#3} % see manual of cleveref, section 8.2.1
\crefformat{subsection}{\S#2#1#3}
\crefformat{subsubsection}{\S#2#1#3}

\theoremstyle{italics}
\newtheorem{theorem}{Theorem}[section]
\newtheorem{lemma}[theorem]{Lemma}
\newtheorem{prop}[theorem]{Proposition}
\newtheorem{cor}[theorem]{Corollary}

\theoremstyle{definition}
\newtheorem{defn}[theorem]{Definition}

\newtheorem{example}[theorem]{Example}
\newtheorem{examples}[theorem]{Examples}

\newtheorem{remark}[theorem]{Remark}
\newtheorem{note[theorem]}{Note}

\usepackage{amsfonts}
\usepackage{colortbl}

\begin{document}

\title{Differential cohomotopy versus  differential cohomology for M-theory
\\
 and differential lifts of Postnikov towers}

\author{Daniel Grady, \; Hisham Sati
}
%\address{Mathematics, Division of Science, New York University Abu Dhabi, UAE.}

\maketitle

\begin{abstract}
  We compare the description of the M-theory form fields via cohomotopy versus that via 
  integral cohomology. The conditions for lifting the latter to the former are identified
  using obstruction theory in the form of Postnikov towers, where torsion plays a central role. 
  A subset of these conditions is shown to correspond compatibly to existing consistency
  conditions, while the rest are new and point to further consistency requirements for M-theory. 
    Bringing in the geometry  leads to  a differential refinement  
    of the Postnikov tower, which should be of independent interest.  
  This provides another confirmation  that cohomotopy is the proper generalized cohomology 
  theory to describe these fields.
\end{abstract}

\vspace{-.6cm}

\tableofcontents

%%%%%%%%%%%%%%%% 
\section{Introduction}
%%%%%%%%%%%%%

One of the main problems that would shed some light on M-theory is the precise 
nature of the C-field in the theory. Earlier literature viewed the C-field 
 as a cocycle in cohomology, or a higher gauge field, with some extra structure 
(see \cite{DFM}\cite{AJ}\cite{tcu}\cite{SSS3}\cite{E8}\cite{FSS1}).
More recently, homotopy theory has been used to describe the dynamics of the 
C-field in M-theory, leading to the proposal in \cite{S-top} that it is quantized 
in \textit{cohomotopy}  cohomology theory $\pi^\bullet$ \cite{Borsuk36}\cite{Spanier49}. 
Later it was shown that  cohomotopy captures the fields 
in M-theory very nicely, for the dynamics in the 
\textit{rational} approximation
\cite{FSS13}\cite{FSS2}\cite{FSS3}\cite{FSS16b}\cite{FSS18}\cite{GaugeEnhancement} \cite{ADE};
see \cite{FSS19} for a review.

\medskip
At a first approximation, ignoring torsion, we have that rational cohomology is essentially equivalent 
to rational cohomotopy in the same degree. For the \textit{stable} case, 
for degree four corresponding to the field strength 4-form $G_4$, 
we have an isomorphism 
%this is the following statement 
\(
H^4(Y^{11}; \Q)  \cong \pi^4_s(Y^{11}) \otimes \Q\;.
\label{eq-rational} 
\)
 %e.g. Hilton 
In the \textit{unstable} case, schematically, we have 
\(
\text{Rational cohomotopy} \;\; \longleftrightarrow  \;\; \text{Rational cohomology} + 
\text{trivialization of the cup square}. 
\label{eq-rational} 
\)
In this case we do not have an isomorphism; for example for $Y^{11}=S^7\times \R^4$,
 we have $H^4(S^7\times \R^4;\Q)=0$, while $\pi^4(S^7\times \R^4)\otimes \Q\cong\Q$.
Hence, it might seem like there is nothing to be gained here by bringing in (co)homotopy 
theory. Nevertheless, somewhat surprisingly, placing the problem in homotopy theory, 
even rationally, has brought in interesting structures beyond just this (as indicated above). 
This interesting rational structure come from the unstable case. Rationally and stably 
$S^4_{\Q}$ is just the Eilenberg-MacLane space $K(\Q,4)$ so there is not much new to say. 
On the other hand, \textit{integrally} and \textit{stably} we do see new effects, which is what we highlight 
 here. 
 
 \medskip
More remarkably, going beyond the rational approximation,  the cancellation 
of the main anomalies of M-theory follows naturally from cohomotopy. It was 
shown in \cite{FSS-twist}\cite{FSS-level} that the C-field charge quantization 
in twisted cohomotopy implies various fundamental anomaly cancellation and 
quantization conditions, with similar effects for D-branes and orientifolds 
\cite{BSS}\cite{SS}. This led to the formulation of: 

\vspace{3mm}
\hypertarget{HypothesisH}{}

\vspace{-3mm} 
\begin{center}
\noindent  \fbox{{\bf Hypothesis H.} {\bf \it  The C-field is charge-quantized
in cohomotopy theory}, even non-rationally. }
\end{center} 

\vspace{-1mm}
\noindent Rational cohomotopy of spacetime $Y$ is given by homotopy classes
of maps to the rational 4-sphere, $[Y, S^4_\Q]$, while cohomotopy deals with 
maps to the  standard 4-sphere, equipped with the usual subspace topology,
$[Y, S^4]$. Since we are a priori given the former, 
we ask for a natural lift to the integral level, and whether this would indeed 
give us the latter. 
This amounts to giving the $\Z$-form for the Sullivan algebra associated with the 
rational homotopy type $S^4_\Q$, as  in \cite[p. 246]{FOT}.
Furthermore, we need to know whether or not the result of the rationalization 
is indeed a finite-dimensional space. This would give us a topological space of the same
rational homotopy type of $S^4_\Q$. 
 As we are ultimately interested in differential refinements, we need to 
have this as a finite-dimensional manifold, i.e., the smooth 4-sphere with its 
standard differentiable structure. 
These generally follow via the result of Sullivan
 on realizability; see \cite[Theorem A]{Sull}\cite{Sull71}\cite{SR}. 

%%%%%%%%%%%%%%%%%%%%
%\paragraph{De-rationalizing the 4-sphere.}  
%%%%%%%%%%%%%%%%%%%%%

\medskip
%Let us consider in more detail the transition from the rational 4-sphere to something non-rational, i.e.,
%integral. The most natural candidate is the actual 4-sphere
%itself. 
If we start with the rational 4-sphere $S^4_\Q$, then how can we lift it to an ``integral"
space? We need to `supply the missing' torsion information that was killed upon the reverse 
process of rationalization. We would get a space ${S}_\Z^4$. What is this?  There could be 
many spaces whose rationalizations coincide; in fact infinitely many, measured by the 
Mislin genus \cite{Mi}. This is the case even if the spaces coincide when localized at every 
prime $p$, not just the trivial prime corresponding to rationalization. 
The \textit{extended genus} \cite{Hi} of a space $X$ is the set of homotopy types $[Y]$ 
of nilpotent CW-spaces $Y$ which are locally homotopy equivalent to $X$ at each prime. 
The \textit{Mislin genus} \cite{Mi} of $X$ is defined to be the subset of the extended
genus consisting of those $[Y]$s of finite type. While the Mislin genus can be finite \cite{Wil},
the extended genus is always infinite for nontrivial homotopy types \cite{Mc1}\cite{Mc2}. 
The corresponding spaces might also not 
be finite-dimensional or of finite type; in fact most of them will not be (see \cite{Wil}).

\medskip
However, the actual 4-sphere $S^4$ stands out as not only the most natural but the 
finite-dimensional one. 
$$
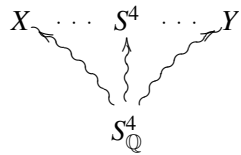
\xymatrix{
X \ar@{}[r]|{\cdot \;\; \cdot \;\; \cdot} & S^4 \ar@{}[r]|{\cdot \;\; \cdot \;\; \cdot} & Y
\\
& S^4_\Q 
\ar@/^-.5pc/@{~>}[ul]
\ar@{~>}[u]
\ar@/^.5pc/@{~>}[ur]
 &
}
$$
Aiming for $S^4$ itself, one has to go through the fibration
 $
 S^4_{\tau} \to S^4 \to S^4_\Q
 $,
 where $S^4_\tau$ is the pure torsion part (see, e.g., \cite{MS}).  This is the source of the torsion 
 obstructions we will identify. 
Now that we have explained what is involved in moving beyond the rational approximation,
we also found that the 4-sphere itself is the most natural. Hence we will adopt this 
perspective henceforth and consider a lift of the form 
 \(
 \label{toplift}
 \xymatrix{
 &&& S^4 \ar[d]
 \\
 Y 
 \ar@{..>}[urrr]^-{
                \mbox{\!\!
        \tiny
        \color{blue}
        \begin{tabular}{c}
         Integral,
          \\
          torsion
            \end{tabular}
      \!}
              }
  \ar[rrr]|-{
                \mbox{\!\!
        \tiny
        \color{blue}
        \begin{tabular}{c}
         Rational,
          \\
         non-torsion
            \end{tabular}
      \!}
              } 
&&&
S^4_\Q
 }
\)
With $S^4$ as the proper lift of the rational homotopy type, 
 the natural question then becomes: before throwing in additional structures, such as 
 differential refinements, how different is the description via (twisted) 
 cohomotopy from the description via (twisted/shifted) integral cohomology? 
 To answer this, we would like to start with integral cohomology as describing the (shifted/twisted) 
 C-field and then transition to a description in terms of cohomotopy. By representability,
 this amounts to lifting 

 \vspace{-4mm} 
 \(
 \label{S4toKZ4}
 \xymatrix{
 &&& S^4 \ar[d]
 \\
 Y 
 \ar@{..>}[urrr]^-{
                \mbox{\!\!
        \tiny
        \color{blue}
        \begin{tabular}{c}
         Nonlinear
          \\
          prequantum
            \end{tabular}
      \!}
              }
  \ar[rrr]|-{
                \mbox{\!\!
        \tiny
        \color{blue}
        \begin{tabular}{c}
         Linear
          \\
          quantum
            \end{tabular}
      \!}
              } 
&&&
K(\Z, 4) \;.
 }
 \)
 The map from the 4-sphere to the Eilenberg-MacLane space $K(\Z, 4)$ 
 assembles, upon taking homotopy  classes, into the integral cohomology 
 $H^4(S^4; \Z)$ generated by a fundamental class.

 \medskip
It turns out that such maps to spheres are 
 quite involved, but they can be seen to arise via an infinite number of intermediate 
 maps,  nonetheless packaged nicely in terms of other Eilenberg-MacLane spaces. 
  The series of approximations of a space by Eilenberg-MacLane spaces,
 and starting with one, assemble into its \textit{Postnikov tower}. The successive liftings from 
 one level to the other is governed by obstruction theory. 
 The 4-sphere admits a Postnikov tower of principal fibrations by virtue of it being 
 simply connected (see \cite{MP}). Note that the Postnikov tower for odd spheres 
 is easier to deal with; see  \cite[Sec. 27.4]{FF}, while 
even spheres are much more involved.

\medskip
Even when we adopt the 4-sphere and start going through the Postnikov tower
and identify the obstructions, there remains a question of what happens beyond the stage
seen by spacetime, beyond which there are infinite number of layers of the sphere, since
the number of nontrivial homotopy groups is countably infinite. That is, we ask:
which spaces look like the 4-sphere through the eyes of the  4th Postnikov section functor?
This can be answered using the notion of a \textit{Postnikov genus}
\cite{MS}. Unlike the case of odd sphere, this is quite involved for even spheres, including 
the 4-sphere. 
%It turns out that there this set is uncountably infinite. 

 \medskip
Notwithstanding the above subtleties, overall what we have is a description of the form 
 \vspace{-1mm} 
 $$
\fbox{ 
$\text{C-field in (twisted)}\; \pi^4(Y^{11}) \; \Longleftrightarrow \; \text{C-field in (twisted)}\; H^4(Y^{11}; \Z) 
 + \text{nontrivial conditions}$. 
 }
 $$
 Explicitly then, one of the main goals of this paper is to unpack and describe these nontrivial conditions
arising from obstruction theory and highlight what they correspond to on the physics side. 

\vspace{-2mm} 
\paragraph{Differential refinement.}
Ultimately we would like refine the topological lift \eqref{toplift} to 
 a geometric lift at the level of smooth stacks of the form 
 \(
 \label{geomlift} 
 \xymatrix@C=4em{
 &&& \widehat{S^4} \ar[d]
 \\
 Y 
 \ar@{..>}[urrr]^-{
                \mbox{\!\!\!\!\!\!\!\!\!\!\!
        \tiny
        \color{blue}
        \begin{tabular}{c}
         Differential cohomotopy, 
          \\
          prequantum and geometric 
            \end{tabular}
      \!}
              }
  \ar[rrr]|-{
                \mbox{\!\!
        \tiny
        \color{blue}
        \begin{tabular}{c}
         Differential cocycle,
         \\
          quantum and geometric 
            \end{tabular}
      \!}
              } 
&&&
\mathbf{B}^3U(1)_\nabla\;.
 }
 \)
 where $\widehat{S^4}$ is the differential refinement of the 4-sphere 
 (see \cite{FSS2}\cite{FSS3}\cite{SS2} for complementary approaches)
 and $\mathbf{B}^3U(1)_\nabla$ is the smooth stack of 3-bundles with 
 connections (see \cite{Cech}\cite{SSS3}\cite{FSS14c}\cite{FSS-cup} \cite{Urs}).
This would require a differential refinement of the Postnikov tower which uses
refinement of cohomology operations, primary  \cite{GS2} and 
secondary  \cite{GS1}. 

\medskip
In between full non-abelian
cohomotopy and abelian ordinary cohomology sits
\textit{stable cohomotopy}, represented not by actual spheres,
but by their stabilization to the sphere spectrum. There is a description of the 
C-field in each one of these flavors (see \cite{FSS-twist}\cite{GaugeEnhancement}):

\begin{center}
  \begin{tabular}{|c||c|c|c|c|}
    \hline
    \begin{tabular}{c}
  \bf    Cohomology
      \\
  \bf    theory
    \end{tabular}
    &
    \begin{tabular}{c}
      Rational
      \\
      cohomology
    \end{tabular}
    &
    \begin{tabular}{c}
      Integral
      \\
      cohomology
    \end{tabular}
    &
    \begin{tabular}{c}
      Stable
      \\
      cohomotopy
    \end{tabular}
    &
    \begin{tabular}{c}
      Non-abelian
      \\
      cohomotopy
    \end{tabular}
    \\
    \hline
    \hline
  \rule{0pt}{2.5ex}  \bf Cocycle
    &
    $G_4$
    &
    $\widetilde G_4$
    &
    $\Sigma^\infty c$
    &
    $c$
    \\
   \hline
  \end{tabular}
\end{center}
 %\noindent ~~~~~~~~~~~~~~~~~~~~~~~~~~~~~~~~~{\bf \footnotesize Table 2 } -- 
%{\it \footnotesize Incremental approximations to full non-abelian Cohomotopy.}
\vspace{-1mm}
We will consider both stable and unstable forms of cohomotopy, with more emphasis on the latter. 
We will also not make a distinction in notation, and use $G_4$ uniformly. 

\medskip
All the conditions we will encounter are torsion in the topological case  (in the stable range, at least), 
with further contribution from refinement of integral classes in the differential case. 
Note that the M-theory fields have been considered from the point of view of Morava 
K-theory $K(n)$ \cite{SW}\cite{SY}, which sits somewhat in between rational cohomology
$H^*(Y^{11}; \Q)\cong K(0)(Y^{11})$ and mod $p$ cohomology 
$H^*(Y^{11}; \Z_p)\cong K(\infty)(Y^{11})$.

\medskip
The paper is organized as follows.
We consider the systematic comparison between the cohomology and cohomotopy
treatments of the C-field in \cref{Sec-EMvsS4}. First we consider 
$\Z_2$ coefficients  in \cref{prime2}, the obstructions classes for which we 
identify in \cref{Sec-ident}, and then consider $\Z_3$ and $\Z_5$ coefficients
in \cref{prime3}. Putting all together gives us the Postnikov tower with 
$\Z$ coefficients in \cref{Sec-Z}.  These stable considerations are then extended to 
 unstable 4-cohomotopy in \cref{Sec-unstable}. Then we describe physical manifestations  
 and corresponding examples in \cref{Sec-Ex}. 
 From the topological case we move to differential refinements in \cref{Sec-diffcoh},
 where we first describe differential cohomotopy in 
\cref{Sec-diffcoh},  then characterize the torsion obstructions in differential cohomology
in \cref{Sec-TorDiff}. This allows us to compare differential cohomotopy and  
differential cohomology in 
\cref{Sec-DiffCohvsDiffCoh}, which serves as a refinement of the topological description in 
\cref{Sec-EMvsS4}, and in which we also provide examples and main applications to M-theory.

%%%%%%%%%%%%%
\section{Cohomological interpretation of cohomotopy: $K(\Z, 4)$ vs. $S^4$}
\label{Sec-EMvsS4}
%%%%%%%%%%%%%%

We consider the comparison between degree four integral cohomology $H^4(Y; \Z)$, 
given by maps from $Y$ to the Eilenberg-MacLane space $K(\Z, 4)$, and degree four 
cohomotopy $\pi^4(Y)$, given by maps to the 4-sphere $S^4$, as in diagram \eqref{S4toKZ4}.
To that end, we will provide a description of cohomotopy via integral and mod $p$ 
cohomology, together with corresponding cohomology operations leading to conditions
on the fields. We will start with mod $p$ coefficients, for $p\in \{2, 3, 5\}$,
 and then assemble into integral coefficients.

%%%%%%%%%%%%%%%
\subsection{$\Z_2$ coefficients } 
\label{prime2}
%%%%%%%%%%%%%%%%

We will use obstruction theory, one dimension at a time, in the range of 
dimensions relevant for M-theory, and extensively applying 
the constructions and presentation of the Postnikov tower in \cite{MT}. 
Since degree $n$ cohomotopy of spaces of dimension less than $n$ is trivial 
by $n$-connectedness of the target sphere (see \cite{Wh}),
we will start with dimension four.

\begin{remark}[Notation on the tower]
  We will be constructing a tower of fibrations $F_n\to X_n\to X_{n-1}$, called {\it stages}. 
  We will use the following notation and conventions for operations on cohomology classes.
  \vspace{-2mm} 
  \begin{itemize}
    \setlength\itemsep{-4pt}
  \item The operation $i^*$ will always denote pullback along the fiber inclusion $F_n\into X_n$. 
  We leave the stage implicit in the notation, however the argument of $i^*$ will make this explicit.
  \item The operation $p^*$ will always denote pullback to the total space, again leaving the fibration implicit.
    \item The operation $\tau$ will always denote the transgression from the base space to the fiber.
  \item Whenever a power operation acts on a class in cohomology which does not have the correct 
  coefficients, this should always be understood as the power operation applied to the appropriate mod p reduction.
    \end{itemize}
  \end{remark}

%%%%%%%%%%%%%%
\paragraph{\underline{Dimension 4:}}
%%%%%%%%%%%%%%%%%%
The space $K(\Z, 4)$ has the same cohomology and homotopy groups as $S^4$ up 
to dimension 4, as $H^4(S^4; \Z) \cong \Z \cong H^4(K(\Z, 4); \Z)$, which in fact
holds with any coefficients.  This means that if $Y$ has dimension 4 then the two descriptions 
agree. In fact, the Hopf degree theorem (see \cite[IX (5.8)]{Kosinski93}) in our case
 states that the 4th cohomotopy classes  
 $[Y \overset{f}{\to} S^4]\in \pi^4(Y)$ of $Y$ are in bijection with 
the \textit{degree} $\deg(f) \in \Z$ of the representing functions, hence that there is a bijection
$$
\xymatrix{{\rm deg}: 
\pi^4(Y):= [Y, S^4] \ar[rr]^-{S^4 \to K(\Z, 4)}_-\simeq && H^4(Y; \Z):= [Y, K(\Z, 4)] \cong \Z
}
$$
from the 4th cohomotopy to the 4th integral cohomology.
This map ${\rm deg}$
is given by ${\rm deg}(f)=f^*([S^4]^*)$, with $[S^4]\in H_4(Y; \Z)$ the fundamental 
homology class of $S^4$ and $[S^4]^*$ is its dual.
%(see, e.g., \cite[Prop. 2.2]{OS}).
We can form factorization $Y \to S^4 \overset{k}{\to} S^4$ as maps from 
 $Y$ to $S^4$ of degree $k \geq 1$. 
Then $[k]=k \iota_4$ in $\pi_4(S^4)\cong \Z$, where $\iota_4=[1]$
is the class of the identity map. In this (and the next) dimension $f^*: \pi^4(S^4) \to \pi^4(Y)$ 
is a homomorphism and we have $\pi^4(S^4) \cong \pi_4(S^4)$ as 
groups, then $[k \circ f]=f^*[k]=kf^*(\iota_4)= k [f]$ in $\pi^4(Y)$. 
See, e.g., \cite{Mil}\cite{OR}. 
 
\medskip
This then captures the essence of cohomotopy for 
spacetimes with only a 4-dimensional manifold piece $X^4$ being topologically 
 nontrivial.

%%%%%%%%%%%%%%
\paragraph{\underline{Dimension 5:}}
%%%%%%%%%%%%%%%%%%

At this first stage we consider the cohomology group in degree five, $H^5(K(\Z, 4); \Z_2)=0$, 
so there is no obstruction in dimension five, which means that if $Y$ is a 5-dimensional 
spacetime, then every degree 4 class lifts on a five manifold, but not uniquely.
The two descriptions still agree in the sense that there are no obstructions coming 
from cohomotopy beyond what we have in cohomology.

\medskip
However, here an interesting effect occurs, analogous to the 
4-dimensional case. 
 For $Y$ of dimension at most five we can use the results of Pontrjagin and Steenrod
 (see, e.g., \cite[Theorem 16.9]{Ba}).
 For $u\in H^4(S^4; \Z)$ a generator, the \textit{degree} map 
 $
 \deg: [Y, S^4] \to H^4(Y; \Z)
 $
 with $\deg(F)=F^*(u)$ is surjective with inverse image 
 \(
  \label{deg5}
 \deg^{-1}(u)\cong H^5(Y; \Z_2)/{\rm Sq}^2_\Z H^3(Y; \Z)\;.
\)
This places conditions on the cohomology of $Y$ and will be relevant in 
the second obstruction in \cref{Sec-ident} and in Remark \ref{Rem-pure5}
and Remark \ref{Rem-purelyco}.

%%%%%%%%%%%%%%
\paragraph{\underline{Dimension 6:}}
%%%%%%%%%%%%%%%%%%

At this stage we enter the stable range.  The cohomology group in degree six, 
$H^6(K(\Z, 4); \Z_2)\cong \Z_2$, generated by ${\rm Sq}^2 \iota_4$, where
$\iota_4 \in H^4(K(\Z, 4); \Z)\cong \Z$ is the fundamental class  
acted upon by the Steenrod square ${\rm Sq}^2\rho_2: K(\Z, 4) \to K(\Z_2, 6)$,
where $\rho_2$ is mod 2 reduction. 
The first stage  of the Postnikov tower is the pullback along ${\rm Sq}^2\rho_2$ of the 
path-loop fibration of the codomain, i.e.,  
$$
\xymatrix{
K(\Z_2, 5)=\Omega K(\Z_2, 6) \ar[r]& X_1 \ar[d] && PK(\Z_2, 6) \ar[d]
\\
& K(\Z, 4) \ar[rr]^{{\rm Sq}^2\rho_2} && K(\Z_2, 6)
}
$$
with $X_1$ being a better approximation to $S^4$ than $K(\Z, 4)$ is.  Indeed, by definition 
the class ${\rm Sq}^2\rho_2 \in H^6(K(\Z, 4); \Z_2)$ has been killed on $X_1$. 
The fundamental class $\iota_5$ of the fiber $K(\Z_2, 5)$ transgresses to 
${\rm Sq}^2\rho_2 \iota_4$.  The cohomology of $X_1$ can be calculated via the 
Serre long exact sequence 
$$
\xymatrix{
H^*(K(\Z_2, 5); \Z_2) \ar[dr]_-{\tau} & H^*(X_1; \Z_2)  \ar[l]_-{\;\;\; i^*}\\
& H^*(K(\Z, 4); \Z_2)\;. \ar[u]_{p^*}
}
$$
The transgression is given by $\tau(\iota_5)={\rm Sq}^2 \rho_2 \iota_4$, which gives in particular that 
$H^6(X_1;\Z_2)=0$ (using the Adem relation ${\rm Sq}^1{\rm Sq}^2={\rm Sq}^3$).

\medskip
Here we have another interesting effect, which is the last dimension 
in our case where cohomotopy is a \textit{group} as opposed to just a set. 
 The set $[Y, S^4]$ has a natural group structure if $Y$ has dimension at most 6.
 This follows from the fact that $S^4$ and the loop space of its suspension 
 $\Omega \Sigma S^4 \simeq \Omega S^5$ have the same
 homotopy 7-type, by the Freudenthal suspension theorem -- see 
 \cite[Chapter 14]{MT}\cite{Wh}.

%%%%%%%%%%%%%%
\paragraph{\underline{Dimension 7:}}
%%%%%%%%%%%%%%%%%%

The next step is to kill $H^7(X_1; \Z_2)$, obtaining a fiber sequence 
$F_2 \to X_2 \to X_1$, with $X_2$ having cohomology in dimension 7, so 
the same homotopy groups as $S^4$ in dimension 6. 
$$
\xymatrix{
K(\Z_2, 6)=\Omega K(\Z_2, 7) \ar[r]& X_2 \ar[d] && PK(\Z_2, 7) \ar[d]
\\
& X_1 \ar[rr]^-{\alpha_7} && K(\Z_2, 7)\;.
}
$$
The transgression vanishes on ${\rm Sq}^2\iota_5\in H^7(K(\Z_2,5);\Z_2)$ and one gets 
$H^7(X_1; \Z_2) \cong \Z_2$ generated by a class $\alpha_7$ such that 
$$
i^*(\alpha_7)={\rm Sq}^2 \iota_5\;,
$$
where $i^*: H^7(X_1; \Z_2) \longrightarrow H^7(K(\Z_2, 5); \Z_2)$ is the map in \eqref{serrespec}. 
Indeed, if $H^7(X_2; \Z_2)=0$ then $X_2$ is an improvement over $X_1$ 
as an approximation to  the 4-sphere $S^4$. 

%%%%%%%%%%%%%%
\paragraph{\underline{Dimensions 8, 9, 10:}}
%%%%%%%%%%%%%%%%%%
At this level, we consider $H^8(X_1; \Z_2)$. There is a class 
$p^*({\rm Sq}^4 \iota_4)$ and also a class $\beta_8$ such that 
$i^*(\beta_8)= {\rm Sq}^3 \iota_5$, which has an indeterminacy 
since $i^*(p^*({\rm Sq}^4 \iota_4))=0$, but the identification 
does not depend on this choice.  
We need to kill the class $p^*({\rm Sq}^4 \iota_4) \in H^8(X_2; \Z_2)$.
This requires working out the Bockstein relations at this level, which is 
done in \cite[Ch. 12]{MT}. The procedure is to map $X_2$ into $K(\Z_8, 8)$
by a map corresponding to a class which reduces to $p^*({\rm Sq}^4 \iota_4)$ (mod 2). This leads to the fibration 
\begin{equation}\label{x3fibration}
\xymatrix@R=1.5em{
K(\Z_8, 7)=\Omega K(\Z_8, 8) \ar[r]& X_3 \ar[d] && PK(\Z_8, 8) \ar[d]
\\
& X_2 \ar[rr]^-{``{\rm Sq}^4 \iota_4"}  && K(\Z_8, 8)\;,
}
\end{equation}
%
%
%\medskip
%Transgression gives $H^*(X_3; \Z_2)$ as in the table ... 
%
where by the quotations $``-"$ we mean that the map is a lift of the mod 2 class $p^*({\rm Sq}^4\iota_4)$. 
Note that this  process kills not only $H^8(X_2;\Z_2)$, but also $H^9(X_2;\Z_2)$ and $H^{10}(X_2;\Z_2)$. 

%$H^8$ but also $H^9$ and $H^{10}$. 

%%%%%%%%%%%%%%
\paragraph{\underline{Dimension 11:}}
%%%%%%%%%%%%%%%%%%

Next we must kill a class $P_{11} \in H^{11}(X_3; \Z_2)$. The class satisfies $i^*P_{11}={\rm Sq}^4 \iota_7$ (see \cite[page 121]{MT}). It is 
the obstruction for the next fibration 
$$
\xymatrix@R=1.5em{
K(\Z_2, 10)=\Omega K(\Z_2, 11) \ar[r]& X_4 \ar[d] && PK(\Z_2, 11) \ar[d]
\\
& X_3 \ar[rr]^-{P_{11}}  && K(\Z_2, 11)\;.
}
$$

%%%%%%%%%%%%%%
\vspace{-3mm} 
\paragraph{\underline{Dimension 12:}}
%%%%%%%%%%%%%%%%%%

The next step is to kill $H^{12}(X_4; \Z_2)$ by using a map 
$$
\xymatrix{
X_4 \ar[rr]^-{``{\rm Sq}^8 \iota_4"} && K(\Z_{16}, 12)\;,
}
$$
where this class has mod 2 reduction equal to ${\rm Sq}^8\iota_4$. But of course, 
${\rm Sq}^8\iota_4=0$ is automatic, from which we can infer at least that $"{\rm Sq}^8\iota_4"$ 
is a multiple of 2 times the generator. Due to the dimension of spacetime, the obstruction at this 
level is irrelevant for the lifting of $G_4$. However, we will see a twelve manifold appear 
when analyzing the Chern-Simons term in the M-theory action and then this condition will become relevant.

%%%%%%%%%%%%%%%%%%%%%%%%
\subsection{Identifying the obstruction classes}
\label{Sec-ident} 
%%%%%%%%%%%%%%%%%%%%%%%%

We identify the relevant mod 2 classes above via transgressions. Note that in some cases (i.e. $\Z_8$ and $\Z_{16}$, 
we need to pass to lifts of the corresponding mod 2 classes (see \cite[Ch. 12]{MT}): 

\vspace{-3mm} 
\begin{itemize}
\item {\it The transgressions: Universally} We consider 
$\tau: H^j({\rm fiber}; \Z_2) \to H^{j+1}({\rm base}; \Z_2)$. 

\begin{itemize}
\vspace{-2mm} 
\item \fbox{$j=4$} $\iota_4 \in H^4(K(\Z, 4); \Z_2)$.

\item \fbox{$j=5$} We have the class $\iota_5 \in H^5(K(\Z_2, 5); \Z_2)$. Here, under transgression 
\newline $H^5(K(\Z_2, 5); \Z_2) \xrightarrow{\tau} H^6(K(\Z, 4); \Z_2)$, we have $\tau(\iota_5)={\rm Sq}^2 \iota_4$.

%\item \fbox{$j=6$} We start with the class  ${\rm Sq}^1\iota_5 \in H^6(K(\Z_2, 5); \Z_2)$. 
%Under transgression $H^6(K(\Z_2, 5); \Z_2) \xrightarrow{\tau} H^7(K(\Z, 4); \Z_2)$, we 
%have $\tau({\rm Sq}^1\iota_5)={\rm Sq}^3\iota_4$.
%
%\item \fbox{$j=7$} We start with the class  ${\rm Sq}^2\iota_5 \in H^7(K(\Z_2, 5); \Z_2)$. 
%Under transgression $H^7(K(\Z_2, 5); \Z_2) \xrightarrow{\tau} H^8(K(\Z, 4); \Z_2)$, we 
%have $\tau({\rm Sq}^2\iota_5)={\rm Sq}^2\tau(\iota_5)={\rm Sq}^2 {\rm Sq}^2 \iota_4=
%{\rm Sq}^3 {\rm Sq}^1 (\iota_4)= {\rm Sq}^3(0)=0$.

\end{itemize}

\vspace{-4mm} 
\item {\it The transgressions: From $(S^4)_1$.} We consider 
$\tau: H^j({\rm fiber}; \Z_2) \to H^{j+1}((S^4)_1; \Z_2)$. 

\begin{itemize}
\vspace{-2mm} 
\item \fbox{$j=8$} We have a class 
${\rm Sq}^4\iota_4 \in H^8((S^4)_1; \Z_2)$ that survives.

\end{itemize} 

\vspace{-4mm} 
\item {\it The transgressions: From $(S^4)_2$.} We consider 
$\tau: H^j({\rm fiber}; \Z_2) \to H^{j+1}((S^4)_2; \Z_2)$. 

\begin{itemize}
\vspace{-2mm} 
\item \fbox{$j=8$} We have a class 
${\rm Sq}^4\iota_4 \in H^8((S^4)_2; \Z_2)$ that survives.

\end{itemize} 

\vspace{-4mm} 
\item {\it The transgressions: From $(S^4)_3$.} We consider 
$\tau: H^j({\rm fiber}; \Z_2) \to H^{j+1}((S^4)_3; \Z_2)$. 

\begin{itemize}
\vspace{-2mm} 
\item \fbox{$j=10$} We have the class $\iota_{10} \in H^{10}(K(\Z_2, 10); \Z_2)$. 
Under transgression \newline $H^{10}(K(\Z_2, 10); \Z_2) \xrightarrow{\tau} H^{11}((S^4)_3; \Z_2)$, 
we have $\tau(\iota_{10})=P_{11}(\iota_4)$.
\end{itemize} 

\vspace{-4mm} 
\item{\it The transgressions: From $(S^4)_4$.} We consider 
$\tau: H^j({\rm fiber}; \Z_2) \to H^{j+1}((S^4)_4; \Z_2)$. 

\begin{itemize}
\vspace{-2mm} 
\item \fbox{$j=11$} We have the class $\iota_{11} \in H^{11}(K(\Z_{16}, 11); \Z_2)$. 
Under the transgression \newline $H^{11}(K(\Z_{16}, 11); \Z_2) \xrightarrow{\tau} H^{12}((S^4)_4; \Z_2)$, 
we have $\tau(\iota_{11})={\rm Sq}^8\iota_4=0$.
\end{itemize} 
\end{itemize}

\begin{lemma}[2-primary Postnikov tower of $S^4$] 
Overall, we have the Postnikov tower
$$
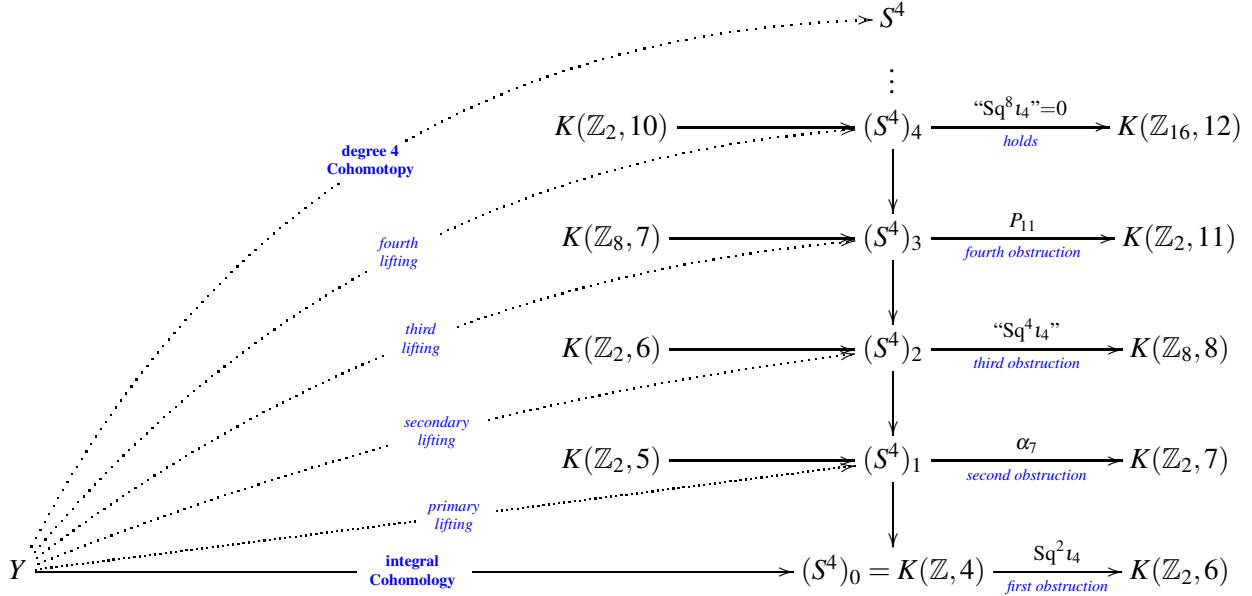
\xymatrix@C=4em{
%K(\Z_{16}, 11) \ar[r]& X_5 \ar[d] & 
&&&&& S^4 \ar@{}[d]^{\!\!\!\!\! \vdots} &
\\
&&&&K(\Z_2, 10) \ar[r] & (S^4)_4 \ar[r]^-{``{\rm Sq}^8 \iota_4"=0}_-{\color{blue}\text{\tiny holds}}  
  \ar[d] & K(\Z_{16}, 12) 
\\
&&&&K(\Z_8, 7) \ar[r] & (S^4)_3 \ar[r]^-{P_{11}}_-{\color{blue}\text{\tiny fourth obstruction}} 
  \ar[d] & K(\Z_{2}, 11) 
\\
&&&&K(\Z_2, 6) \ar[r] & (S^4)_2 \ar[r]^-{``{\rm Sq}^4 \iota_4"}_-{\color{blue}\text{\tiny third obstruction}} 
  \ar[d] & K(\Z_{8}, 8) 
\\
&&&&K(\Z_2, 5) \ar[r] & (S^4)_1 \ar[r]^-{\alpha_7}_-{\color{blue}\text{\tiny second obstruction}} 
  \ar[d] & K(\Z_{2}, 7) 
\\
Y \ar[rrrrr]|-{
                \mbox{\!\!
        \tiny \bf 
        \color{blue}
        \begin{tabular}{c}
         integral
          \\
           Cohomology
            \end{tabular}
      \!}
              }
\ar@{..>}[urrrrr]|{
                \mbox{\!\!
        \tiny 
        \color{blue}
        \begin{tabular}{c}
         primary
          \\
           lifting
            \end{tabular}
      \!}
              }
\ar@/^1pc/@{..>}[uurrrrr]|-{
                \mbox{\!\!
        \tiny
        \color{blue}
        \begin{tabular}{c}
         secondary
          \\
          lifting
            \end{tabular}
      \!}
              }
\ar@/^2.3pc/@{..>}[uuurrrrr]|-{
                \mbox{\!\!
        \tiny
        \color{blue}
        \begin{tabular}{c}
         third
          \\
           lifting
            \end{tabular}
      \!}
              }
\ar@/^3.5pc/@{..>}[uuuurrrrr]|-{
                \mbox{\!\!
        \tiny
        \color{blue}
        \begin{tabular}{c}
         fourth
          \\
           lifting
            \end{tabular}
      \!}
              }
\ar@/^5pc/@{..>}[uuuuurrrrr]|-{
                \mbox{\!\!
        \tiny \bf 
        \color{blue}
        \begin{tabular}{c}
         degree 4
          \\
           Cohomotopy
            \end{tabular}
      \!}
              }
&&&&& (S^4)_0=K(\Z, 4) \ar[r]^-{{\rm Sq}^2\iota_4}_-{\color{blue}\text{\tiny first obstruction}} 
 & K(\Z_{2}, 6) 
}
$$
\end{lemma} 
This shows that schematically, as in the Introduction,  
$$
``\text{Cohomotopy in degree 4} = \text{Integral 4-cohomology} + \text{four obstructions}".
$$
The main point is to identify the four obstructions above with conditions arising 
from M-theory and provide interpretations  for them. This is done after pulling back to 
spacetime $Y$, where the fundamental class $\iota_4$ pulls back to the field 
$$
G_4 - \tfrac{1}{2}\lambda =:\widetilde{G}_4= f^*\iota_4 
$$
 where $\lambda=\tfrac{1}{2}p_1$ is the first Spin characteristic class of the 
 lifted tangent bundle to spacetime. Here we are starting with the shift of $G_4$, which is quantized. 
 However, this shift itself can be implemented using twisted cohomotopy, as discussed extensively in \cite{FSS-twist}.
 
\paragraph{(i) The first obstruction.}\label{stg1p2}
The first obstruction is 
\footnote{By ``$\overset{!}{=}$", we mean that we require the equality, but also that this 
does not necessarily automatically hold.} 
$$
{\rm Sq}^2\widetilde{G}_4 \overset{!}{=}0 \in H^6(Y; \Z_2)\;.
$$
Indeed, it was shown in \cite{FSS-twist} that this follows from anomaly cancellation 
in M-theory. Note that it is stronger than the obstruction given by the 3rd Steenrod 
square ${\rm Sq}^3$, as the  former is a condition for 
KO-theory while the latter is for K-theory.

\paragraph{(ii) The second obstruction.}\label{stg2p2}
Overall, the second obstruction is
$$f^*(\alpha_7)\overset{!}{=} 0\in H^7(Y;\Z_2)$$
where $\alpha_7$ is a secondary operation, restricting fiberwise to ${\rm Sq}^2\iota_5$. In particular, this means that if $G_4$ \textit{vanishes} in cohomology, we have the more relatable condition
$$
f^*(i^* \alpha_7)= f^*({\rm Sq}^2\iota_5) =
{\rm Sq}^2f^*(\iota_5)\overset{!}{=} 0
\in H^7(Y; \Z_2)\;.
$$
In this case, we impose the condition 
${\rm Sq}^2f^*(\iota_5)=0$. 
At this stage we note that in the current formulation 
there are no fields of degree five in M-theory. Hence we 
will instead impose the natural condition
$$
f^*(\iota_5)=: G_5=0 \in H^7(Y; \Z_2)\;.
$$
 Note that rationally we could consider the possibility that $G_5$ as being 
$*_{11}C_6$, the Hodge dual of the potential for $G_7$. However, there is 
no natural degree four potential, except if we view $G_4$ as such, but this
is closed, hence such a candidate $G_5$ would vanish even as  a form. 
Another possibility is that in the presence of M5-branes, the Bianchi 
identity $dG_4=0$ is violated by a delta function supported on the 
M5-brane worldvolume. The latter can be viewed as giving rise to 
a degree five class, but it does not satisfy the right condition, which is 
torsion. We will get back to this in Remark \ref{Rem-pure5}
and Remark \ref{Rem-purelyco}.

\paragraph{(iii) The third obstruction.}
The third obstruction is 
$$
f^*(``{\rm Sq}^4\iota_4")\overset{!}{=}0 \in H^8(Y; \Z_8)\;.
$$
Note that, by construction, this implies also that (upon mod 2 reduction)
$$
f^*({\rm Sq}^4\iota_4)=
{\rm Sq}^4 f^*(\iota_4)=
{\rm Sq}^4 \widetilde{G}_4 = \widetilde{G}_4\cup \widetilde{G}_4 =0 \in H^8(Y; \Z_2)\;.
$$
This captures the anomaly given by trivializing the cup product of $G_4$,  at least mod 2. 
This is the affect in the stable setting, but (as we will see) the unstable obstruction implies 
the stronger condition that the cup product vanish even integrally.  
Recall that rationally we have the equation of motion 
$$
d*G^{\rm form}_4=\tfrac{1}{2}G^{\rm form}_4 \wedge G^{\rm form}_4\;. 
$$ 
At the level of cohomology classes (with torsion), we have that the cup product of 
the class $G_4$ of $G^{\rm form}_4$ with itself is zero. This of course 
implies immediately that the mod 2 reduction also vanish 
$$
{\rm Sq}^4\widetilde{G}_4=\widetilde{G}_4 \cup \widetilde{G}_4=0\;. 
$$
What about the coefficients being $\Z_8$ rather than $\Z_2$? 
We first argue that $\Z_8$ coefficients are somewhat natural to appear in this context. 
We consider the fields reduced modulo 4,
for instance the shift in the field is given by $\tfrac{1}{2}\lambda$ , where 
$\lambda=\tfrac{1}{2}p_1$ the Spin characteristic class arising from the 
first Pontrjagin class $p_1$ being even in the cohomology of BSpin. If we 
start with an oriented --  rather than a Spin -- setting, then we are considering 
modding out $p_1$ by 4. Hence it makes sense to consider a corresponding class 
$x_4 \in H^4(Y^{11}; \Z_4)$, given by mod 4 reduction, or in the lift to the bounding manifold $Z^{12}$. 
There is an operation $\mathfrak{P}_2:H^4(Y^{11};\Z_4)\to H^8(Y^{11};\Z_8)$ 
called the \textit{Pontrjagin square} operation. It has the property that 
\begin{align}
\mathfrak{P}_2\rho_2(x_4)&=x_4^2  \nonumber
\\
\rho_4\mathfrak{P}_2(x_4)&=x_4^2 \nonumber
\end{align}
  where $\rho_2$ and $\rho_4$ are the mod 2 and 4 reductions, respectively. 
  Then in fact
$$
\rho_2\mathfrak{P}_2(x_4)=\rho_2(x_4)^2={\rm Sq}^4\rho_2(x_4)\;.
$$
This works for any degree 4 class, not just the reduction of $p_1$. Hence the operation
\footnote{We are tempted to identify this operation explicitly as  ``${\rm Sq}^4\iota_4$", 
but unfortunately the Pontrjagin square is an unstable operation, while  ``${\rm Sq}^4\iota_4$" 
is stable. Thus at best there is a stable operation which reduces to the Pontrjagin square in the 
given degree.} 
$\mathfrak{P}_2\rho_4\iota_4$ indeed gives a mod 8 lift of ${\rm Sq}^4\iota_4$.

\paragraph{{\bf (iv)} The fourth obstruction.}
The fourth obstruction is 
$$
f^*(P_{11})\overset{!}{=} 0 \;,
$$
where $P_{11}$ is a class which fiberwise restricts to ${\rm Sq}^4\iota_7$. 
The Universal Coefficient Theorem gives
$$
\xymatrix{
0 \ar[r] & {\rm Ext}_{\Z_2}^1 (H_{11}(Y; \Z_2), \Z_2) 
\ar[r] & H^{11}(Y; \Z_2) \ar[r] & {\rm Hom}(H_{11}(Y; \Z_2), \Z_2) 
\ar[r] & 0
}.
$$
Since $\Z_2$ is a field, the Ext term vanishes and we have the 
isomorphism  $H^{11}(Y; \Z_2)\cong H_{11}(Y; \Z_2)$.
If $Y$ is non-orientable, then this group is trivial, so that 
$P_{11}$ has no effect. However, if $Y$ is 
orientable, then $H_{11}(Y; \Z_2) \cong \Z_2$, so that  we
get a detectable effect for M-theory on orientable spacetimes. 

\begin{remark}[Obstructions as $n$-ary constraints]
The 2-primary Postnikov resolution of the sphere 
can also be organized into primary, secondary, etc. obstructions, in the sense of 
cohomology operations  (see \cite{LT}). 
For our case of the 4-sphere, up to degree eight, it looks as follows (necessarily mixing dimensions):
$$
\xymatrix@C=4em{
K(\Z_2, 7) \ar[r]^-{j_3} &  (S^4)_3 \ar[d] &
\\
K(\Z_2, 6) \times K(\Z_2, 7) \ar[r]^-{j_2} &  (S^4)_2 
\ar[d]  \ar[r]^-{\beta_4}_-{\tiny \color{blue} \rm Tertiary} & K(\Z_2, 8)
\\
K(\Z_2, 5) \times K(\Z_2, 7) \ar[r]^-{j_1} &  (S^4)_1 
\ar[d]  \ar[r]^-{(\alpha_3, \alpha_4)}_-{\tiny \color{blue} \rm Secondary} 
& K(\Z_2, 7) \times K(\Z_2, 8)
\\
&  (S^4)_0=K(\Z, 4)   
\ar[r]^-{({\rm Sq}^2, {\rm Sq}^4)}_-{\tiny \color{blue} \rm Primary} & 
K(\Z_2, 6) \times K(\Z_2, 8)
}
$$
\item {\bf (i) Primary obstruction}:
The Steenrod squares are primary cohomology operations,

\item  {\bf (ii) Secondary obstruction}:
The classes  $(\alpha_3, \alpha_4)$ 
represent secondary cohomology operations, 
\begin{align*}
j_1^* \alpha_3 &= {\rm Sq}^2 \iota_5 \otimes 1\;,
\\
j_1^* \alpha_4 &= {\rm Sq}^2{\rm Sq}^1 \iota_5 \otimes 1 + 1 \otimes {\rm Sq}^1 \iota_7\;.
\end{align*}

\item  {\bf (iii) Tertiary obstruction}:
The class $\beta_4$ represents a tertiary 
cohomology operations
$$
j_2^* \beta_4 = {\rm Sq}^2\iota_6 \otimes 1 + 1 \otimes {\rm Sq}^1 \iota_7\;.
$$

%The 4-stem $(S^4)_4$ and 5-stem $(S^4)_5$ trivial. So for $Y$ of dimension $\leq ??$
%$$
%\pi^4(Y)/({\rm odd\;torsion}) \simeq [Y, (S^4)_3]\;.
%$$
%From here on they assume $n>5$ for $S^n$ -- why????
\end{remark}

%%%%%%%%%%%%%%%
\subsection{$\Z_3$ and $\Z_5$ coefficients} 
\label{prime3}
%%%%%%%%%%%%%%%%

The main contribution to the Postnikov tower of $S^4$ is from the prime $p=2$ as we
saw above. However, the primes 3 and 5 also contribute, albeit to a lesser extent. 
The structure of homotopy groups of spheres give some immediate consequences for 
the Postnikov tower at different primes 3 and 5. In particular, the Serre spectral 
sequence implies that at the second stage we have an isomorphism 
$$
H^*((S^4)_2,\Z_3)\cong H^*((S^4)_1,\Z_3)\cong H^*(K(\Z,4);\Z_3)\;,
$$
and similarly at the prime 5 we have an isomorphism 
$$
H^*((S^4)_4;\Z_5)\cong H^*(K(\Z,4);\Z_5)\;.
$$
For $p$ odd, the structure of mod $p$ cohomology rings of Eilenberg-MacLane spaces was
determined by Cartan \cite{Ca} and Serre \cite{Se}
 in terms of admissible monomials of
 Steenrod reduced powers and the Bockstein (see also \cite[Lecture 30]{FF}).
This has been recast by Tamanoi \cite[Sec. 5.2]{Ta} using the  Milnor basis of the dual Steenrod algebra, giving  
explicit polynomial generators  for $H^*(K(\Z, n); \Z_p)$. Using this identification, we have 
$$
H^8(K(\Z,4);\Z_3)\cong \Z_3\langle \mathcal{P}_3^1\rho_3\iota_4\rangle, ~~~~ H^{12}(K(\Z,4);\Z_3)\cong \Z_3\langle (\rho_3\iota_4)^3\rangle, ~~~~ H^{12}(K(\Z,4);\Z_5)\cong \Z_5\langle \mathcal{P}_5^1\rho_5\iota_4\rangle\;,
$$
where $\mathcal{P}_p^n:H^i(-; \Z_p) \to H^{i+2n(p-1)}(-; \Z_p)$ is the Steenrod reduced power operation 
of degree $n$ at the prime $p$.
Hence we will get conditions on the vanishing of the pullback of the classes 
$$
 \mathcal{P}_3^1\rho_3\iota_4, ~~~~ (\rho_3\iota_4)^3, ~~~~  \mathcal{P}_5^1\rho_5\iota_4\;.
$$

\begin{remark}[Interpretation] 
Mod 3 reductions are shown to play a prominent role in topological considerations in M-theory 
\cite{KSpin}, where  similar conditions, including $ \mathcal{P}_3^1\rho_3G_4=0$, 
have been highlighted in the context of Spin K-theory
\end{remark}

%%%%%%%%%%%%%%%
\subsection{$\Z$ coefficients} 
\label{Sec-Z}
%%%%%%%%%%%%%%%%

Using our discussion in \cref{prime2} 
%%and \cref{prime3} 
we can assemble the tower integrally in
 the desired range. To do this, we observe that by killing all cohomology classes in $H^{n+*}(K(\Z,4+n);\Z_p)$, 
 for fixed $n\gg 11$ and for each prime $p$, we can utilize \cite[Theorem 4, Ch. 10]{MT} to construct 
 a space $(S^4)_4$ for which there exists a map $f:(S^4)_4\to S^4$ inducing an isomorphism on 
 each $p$-component of $\pi_{n+i}$. Since the homotopy groups of $S^4$ in this range are all 
 finitely generated and torsion, this will imply that $f$ is actually 12-connected. 
 Overall, we have the following:

\begin{lemma}[Integral Postnikov tower for $S^4$]
\label{Lem-IntPost}
The tower takes the form
$$
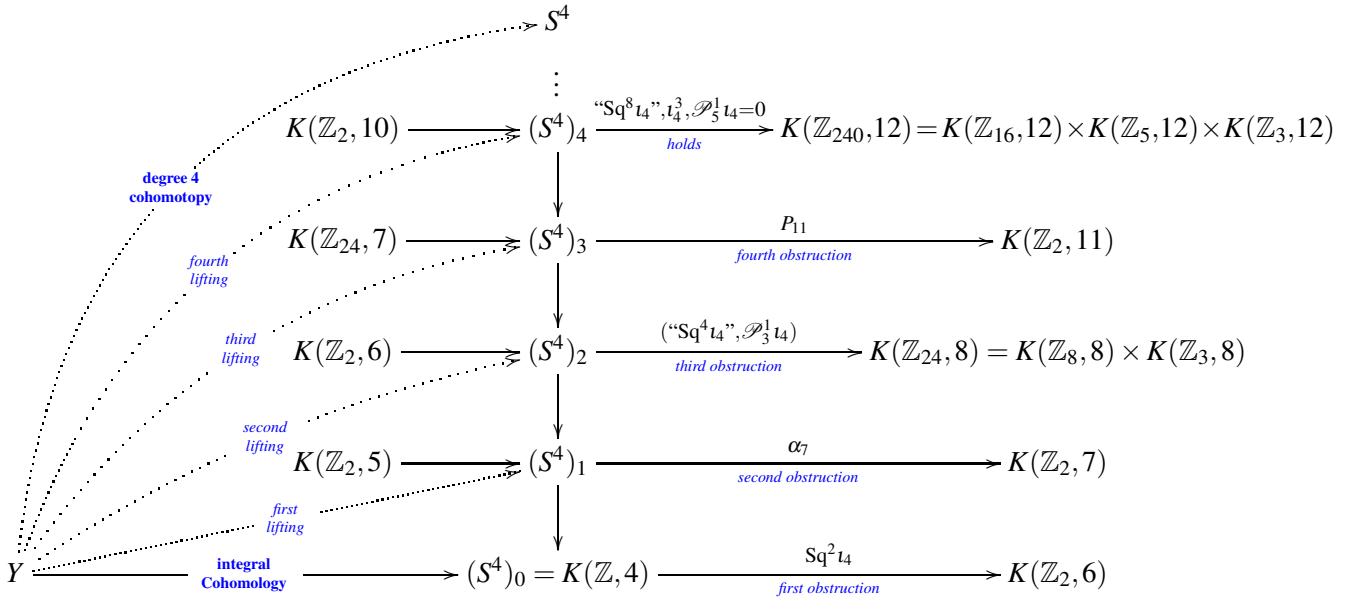
\hspace{-1mm} 
\xymatrix@C=1.7em{
%K(\Z_{16}, 11) \ar[r]& X_5 \ar[d] & 
&&&&& S^4 \ar@{}[d]^{\!\!\!\!\! \vdots} &
\\
&&&&
K(\Z_2, 10) \ar[r] & (S^4)_4 
\ar[rr]^-{\!\!\!\!``{\rm Sq}^8 \iota_4",\iota_4^3,\mathcal{P}^1_5\iota_4=0}_-{\color{blue}\text{\tiny holds}}  
  \ar[d] && K(\Z_{240}, \! 12)\!=\!K(\Z_{16}, \! 12) \! \times \! 
  K(\Z_{5}, \! 12)\!\times \! K(\Z_{3},\! 12) 
\\
&&&&K(\Z_{24}, 7) \ar[r] & (S^4)_3 \ar[rr]^-{P_{11}}_-{\color{blue}\text{\tiny fourth obstruction}} 
  \ar[d] && K(\Z_{2}, 11) 
\\
&&&&K(\Z_2, 6) \ar[r] & (S^4)_2 \ar[rr]^-{(``{\rm Sq}^4\iota_4",\mathcal{P}_3^1\iota_4)}_-{\color{blue}\text{\tiny third obstruction}} 
  \ar[d] && K(\Z_{24}, 8)=K(\Z_{8},8)\times K(\Z_3,8)
\\
&&&&K(\Z_2, 5) \ar[r] & (S^4)_1 \ar[rr]^-{\alpha_7}_-{\color{blue}\text{\tiny second obstruction}} 
  \ar[d] && K(\Z_{2}, 7) 
\\
Y \ar[rrrrr]|-{
                \mbox{\!\!
        \tiny \bf 
        \color{blue}
        \begin{tabular}{c}
         integral
          \\
           Cohomology
            \end{tabular}
      \!}
              }
\ar@{..>}[urrrrr]|{
                \mbox{\!\!
        \tiny 
        \color{blue}
        \begin{tabular}{c}
         first
          \\
           lifting
            \end{tabular}
      \!}
              }
\ar@/^1pc/@{..>}[uurrrrr]|-{
                \mbox{\!\!
        \tiny
        \color{blue}
        \begin{tabular}{c}
         second
          \\
          lifting
            \end{tabular}
      \!}
              }
\ar@/^2.3pc/@{..>}[uuurrrrr]|-{
                \mbox{\!\!
        \tiny
        \color{blue}
        \begin{tabular}{c}
         third
          \\
           lifting
            \end{tabular}
      \!}
              }
\ar@/^3.5pc/@{..>}[uuuurrrrr]|-{
                \mbox{\!\!
        \tiny
        \color{blue}
        \begin{tabular}{c}
         fourth
          \\
           lifting
            \end{tabular}
      \!}
              }
\ar@/^5pc/@{..>}[uuuuurrrrr]|-{
                \mbox{\!\!
        \tiny \bf 
        \color{blue}
        \begin{tabular}{c}
         degree 4
          \\
           cohomotopy
            \end{tabular}
      \!}
              }
&&&&& (S^4)_0=K(\Z, 4) \ar[rr]^-{{\rm Sq}^2\iota_4}_-{\color{blue}\text{\tiny first obstruction}} 
 && K(\Z_{2}, 6) 
}
$$
Note that at the top level the three conditions vanish necessarily on $Y^{11}$, for dimension reasons. 
\end{lemma}

\paragraph{Low-dimensional obstructions.} The obstructions in degree $6$ and $7$ are 
identified with the obstructions at the prime 2, given in \cref{stg1p2}.

\vspace{-4mm} 

\paragraph{The tertiary obstruction.} Here we have again the prime 2 obstructions, identified in \cref{stg1p2}, but also
 a new condition which occurs at the prime 3. Namely, we have the condition 
$$
\mathcal{P}^1_3(\widetilde{G}_4)=0\in H^8(Y;\Z_3)\;.
$$
As indicated above, this is compatible with \cite{KSpin}, where a similar condition was 
proposed using Spin K-theory. 

\vspace{-4mm} 

\paragraph{The quaternary obstructions.} This is identified as the obstruction class $P_{11}$ 
at the prime 2, as in \cref{stg1p2}.

\vspace{-4mm} 
\paragraph{The quinary obstructions.} These obstructions necessarily vanish on $Y^{11}$. 
However we will consider a closed 12-manifold $Z^{12}$ in analyzing the congruences of the 
Chern-Simons term in the M-theory action. In this case, the three conditions 
$$
``{\rm Sq}^8\iota_4\overset{!}{=} 0, \qquad
 \iota_4^3\overset{!}{=}0, 
 \qquad 
  \mathcal{P}^1_5\iota_4\overset{!}{=}0 
$$
are nontrivial. We will see that the second obstruction gives exactly the mod 3 congruences in the M-theory 
action discussed in \cite{flux}.

%From the identification of homotopy groups of spheres in low degrees, it is easy to write down the form of the Postnikov tower integrally and unstably up to the 11-th stage. The tower takes the rather daunting form
%$$
%\xymatrix@C=3em{
%K(\Z/15,11)\ar[r] & (S^4)_7\ar[d]
%\\
%K(\Z/24\times \Z/3,10)\ar[r] & (S^4)_6\ar[d]\ar[r] & K(\Z/15,12)
%\\
%K((\Z/2)^2,9)\ar[r] & (S^4)_5\ar[d]\ar[r] &K(\Z/24\times \Z/3,11)
%\\
%K((\Z/2)^2,8)\ar[r] & (S^4)_4\ar[d]\ar[r] & K((\Z/2)^2,10)
%\\
%  K(\Z\times \Z/12,7)\ar[r] &(S^4)_3 \ar[d] \ar[r] & K((\Z/2)^2,9)
%\\
%   K(\Z/2,6)\ar[r] & (S^4)_2 \ar[d]\ar[r]  &   K(\Z\times \Z/12,8)
%\\
%  K(\Z/2,5)\ar[r] & (S^4)_1 \ar[d]\ar[r]  &  K(\Z/2,7)
%   \\
%  & (S^4)_0=K(\Z, 4)  \ar[r] & K(\Z/2,6)
%}
%$$
%Stably, the tower simplifies considerably. Moreover, the identifications obtained at the primes 2 and 3 allow us to explicitly identify the k-invariants. 
%$$
%\xymatrix@C=3em{
%&\vdots\ar[d] &
%\\
%K(\Z/2,10)\ar[r] & (S^4)_4\ar[d]\ar[r] & K(\Z/240,12)
%\\
%  K(\Z/24,7)\ar[r] &(S^4)_3 \ar[d] \ar[r] & K(\Z/2,11)
%\\
%   K(\Z/2,6)\ar[r] & (S^4)_2 \ar[d]\ar[r]  &   K(\Z/24,8)
%\\
%  K(\Z/2,5)\ar[r] & (S^4)_1 \ar[d]\ar[r]^-{\alpha_7}  &  K(\Z/2,7)
%   \\
%  & (S^4)_0=K(\Z, 4)  \ar[r]^{Sq^2\iota_4} & K(\Z/2,6)
%}
%$$
%

\medskip
Summarizing, we have the following.
\begin{prop}[Cohomotopy vs. integral cohomology] \label{cohvint}
Let $Y^{11}$ be an 11-dimensional (smooth) manifold. 
Then a class $c\in H^4(Y^{11};\Z)$ lifts to a class $\tilde{c}\in \pi^4(Y^{11})$ 
if and only if the following conditions are satisfied.
\begin{enumerate}[{\bf (i)}]
\vspace{-2mm} 
\item $Sq^2(c)\equiv 0 \mod 2$, ~~$\mathcal{P}^1_3(c)\equiv 0 \mod 3$.
\vspace{-2mm} 
\item There is a lift $c^{\prime}:Y^{11}\to (S^4)_1$ of $c$ such that $\alpha_7(c^{\prime})
\equiv 0 \mod 2$.
\vspace{-2mm} 
\item There is a further lift $c^{\prime\prime}:Y^{11}\to (S^4)_2$ such that $\beta_8(c^{\prime\prime})
\equiv 0 \mod 8$. In particular, upon mod 2-reduction, we have $Sq^4(c)=c^2\equiv 0 \mod 2$.
\vspace{-2mm} 
\item There is a further lift $c^{\prime\prime\prime}$ of $c^{\prime\prime}$ 
such that $P_{11}(c^{\prime\prime\prime})\equiv 0 \mod 2$. In particular upon mod 2 reduction, we have the 
tautological relation $Sq^8(c)\equiv 0 \mod 2$.
\end{enumerate}
\end{prop}

Much of the information in the above proposition is 2-torsion. We now directly apply this to our field.  

\begin{prop}
[Cohomotopy vs. cohomology for the C-field] 
\label{cor-C}
 Consider the M-theory (shifted) C-field $\widetilde{G}_4$ as an 
integral cohomology class in degree four. Then if $\widetilde{G}_4$ lifts to 
a cohomotopy class ${\cal G}_4\in \pi^4(Y^{11})$ the following obstructions necessarily vanish 
\begin{enumerate}[{\bf (i)}] 
\vspace{-2mm}
\item ${\rm Sq}^2 \widetilde{G}_4=0 \in H^6(Y^{11}; \Z_2)$. 
\vspace{-2mm}
\item $\mathcal{P}^1_3(\widetilde{G}_4)=0 \in H^8(Y^{11}; \Z_3)$.
\vspace{-2mm}
\item ${\rm Sq}^4 \widetilde{G}_4=\widetilde{G}_4\cup\widetilde{G}_4=0  \in H^{8}(Y^{11}; \Z_{2})$.
\vspace{-2mm}
\item 
If $G^{\rm form}_4=0$ and $dC^{\rm form}_3=0$, so $C_3$ can be lifted to an integral 
class $\widetilde{C}_3$, then we also have 
${\rm Sq}^3{\rm Sq}^1\widetilde{C}_3=0\in H^7(Y^{11};\Z_2)$.
\vspace{-2mm}
\item 
If $dG^{\rm form}_7=G^{\rm form}_4\wedge G^{\rm form}_4=0$ and $G^{\rm form}_7$ 
can be lifted to an integral class $\widetilde{G}_7$, then we also have the condition 
 ${\rm Sq}^4\widetilde{G}_7=0\in H^{11}(Y^{11};\Z_2)$.
\end{enumerate}
%When $Y$ is of dimension less than or equal to 12 then the above are the only obstructions
%at the prime 2. 
\end{prop} 
\begin{proof}
The first three conditions are immediate consequences of Proposition \ref{cohvint}. By stability 
of Steenrod squares, applying the based loop functor to the mapping 
$$
{\rm Sq}^2\iota_4:K(\Z,4)\longrightarrow K(\Z_2,6)
$$
gives ${\rm Sq}^2\iota_3:K(\Z,3)\to K(\Z_2,5)$. Trivializing $G_4=0$ by $\widetilde{C}_3$, 
we get a choice of lift of $\widetilde{G}_4=0$ to the fiber $K(\Z_2,5)$, given by ${\rm Sq}^2\widetilde{C}_3$. 
Then in this case the obstruction class is 
$$
i^*\alpha_7({\rm Sq}^2\widetilde{C}_3 )=
{\rm Sq}^2{\rm Sq}^2\widetilde{C}_3 ={\rm Sq}^3{\rm Sq}^1\widetilde{C}_3\;,
$$
where we have used the Adem relation ${\rm Sq}^2{\rm Sq}^2={\rm Sq}^3{\rm Sq}^1$. 
For the last identification, we identify $\widetilde{G}_7$, after reducing mod 24, with a lift to the fiber $K(\Z_{24},7)$ in the fibration $K(\Z_{24},7)\to (S^4)_3\to (S^4)_2$. Since $P_{11}$ restricts fiberwise to ${\rm Sq}^4\iota_7$, we have the obstruction class 
$$
i^*P_{11}(\widetilde{G_7})={\rm Sq}^4\widetilde{G}_7\;.
$$
 Hence the result follows.   
 \hfill \end{proof}

Proposition \ref{cohvint} also has some immediate 
striking consequences. In particular, Proposition \ref{cor-C} implies that even if $G_4=0$, there are 
still obstructions to lifting the $C$-field to a cohomotopy class. Thus, quantization in cohomotopy 
seems to uncover extremely subtle quantization conditions on the $C$-field. 

\begin{remark}
[Cohomotopy first contribution to the $C$-field]
\label{Rem-pure5}
We highlight that even if $\widetilde{G}_4=0$, there are still obstructions to lifting the $C$-field 
to a class in cohomotopy. In particular, we have a mysterious degree 5 class $\eta\in H^5(Y^{11};\Z_2)$ 
which transgresses to ${\rm Sq}^2\widetilde{G}_4$. By construction of the transgression, this class can 
be interpreted concretely as follows. 
Fix a map ${\rm Sq}^2:K(\Z,4)\to K(\Z_2,6)$ representing the Steenrod square ${\rm Sq}^2$. 
Since $\widetilde{G}_4$ vanishes in integral cohomology, we have a global trivialization 
$\delta \widetilde{C}_3=\widetilde{G}_4$ of $\widetilde{G}_4$ as an integral cochain, 
which gives rise to a trivialization ${\rm Sq}^2\widetilde{C}_3$ of ${\rm Sq}^2\widetilde{G}_4$ 
in $\Z_2$-cohomology, by naturality. Let us fix another trivialization 
$\delta\epsilon={\rm Sq}^2\widetilde{G}_4$ in cochains with $\Z_2$-coefficients. 
Setting $\eta:={\rm Sq}^2\widetilde{G}_4$, we have  
$$
\delta(\eta)=\delta ({\rm Sq}^2\widetilde{C}_3-\epsilon)=
{\rm Sq}^2\widetilde{G}_4- {\rm Sq}^2\widetilde{G}_4=0\;,
$$
so that $\eta$
indeed represents a degree 5 cocycle in $\Z_2$ cohomology, which may be generally nonvanishing.
Note that there is a degree five class associated with the  C-field, namely the fifth integral Steifel-Whitney 
class $W_5$ (see \cite{DFM}), but it is different from this class. 
\end{remark}

\begin{remark}
[Congruences for the M-theory action via cohomotopy]  
\label{congM-pat}
Another interesting effect occurs when considering the Chern-Simons term in the M-theory action
$$
\tfrac{1}{6}\int_{Y^{11}}C_3^{\rm form}\wedge G_4^{\rm form} \wedge G_4^{\rm form}\;.
$$ 
As usual, since $C_3$ may not be globally defined in general, one may consider $Y^{11}$ as the boundary 
of a 12-manifold $Z^{12}$ and analyzes the globally well defined term 
\(
\label{csintwit}
\tfrac{1}{6}\int_{Z^{12}}G_4^{\rm form}\wedge G_4^{\rm form}\wedge G_4^{\rm form}
\)
on $Z^{12}$. To show that the integral is independent of the choice of $Z^{12}$, one considers 
another ${Z}^{\prime 12}$ with boundary $Y^{11}$ and integrates over the closed 12-manifold 
$Q= Z^{\prime 12}\sqcup Z^{12}$. However, as remarked in \cite{Wi}, the usual quantization 
law of $G_4$ does not give rise to a well defined Chern-Simons action, as \eqref{csintwit} might 
fail to be integral by a factor of $6$. 

Note that our obstruction theory works just as well for a closed 12-manifold $Z^{12}$. 
In this case, the obstruction at the top stage of the tower gives the condition
$$
\widetilde{G}_4^3\equiv 0\mod 3\;.
$$
This is in addition to condition {\bf (iii)} in Prop. \ref{cor-C}, which states that 
$\widetilde{G}_4^2={\rm Sq}^4(\widetilde{G}_4)\equiv 0 \mod 2$. 
These two conditions together imply \footnote{Note that divisibility by 2 is not immediate, but can be deduced using the same argument in \cite[p. 12]{flux}, still without reference to an $E_8$-theory.} the divisibility by $6$ condition on $G_4^3$. 
 The crucial distinction is that 
 our congruences are obtained without reference to $E_8$-gauge theory. An alternative formulation of the 
 congruence via a proposed higher form of index theory is given in \cite{Sa1}\cite{Sa2}. 
\end{remark}

%%%%%%%%%%%%%%%%%%%%%%%%%%%%%
\subsection{Obstructions for unstable 4-cohomotopy}
\label{Sec-unstable} 
%%%%%%%%%%%%%%%%%%%%%%%%%%%%%

So far, our work has been limited to the stable range of cohomotopy in degree 4. In part, this is due to the fact 
that the obstruction theory in the unstable case is considerably more complicated. Moreover, working out the 
$k$-invariants in the Postnikov tower unstably does not yield information which can be directly compared 
with existing literature: there are many secondary and tertiary obstructions, which arise as classes defined 
modulo some ambiguity, but are not familiar primary obstructions or Massey products. Nevertheless, we 
highlight the following. 

\begin{remark}[Quaternionic Hopf fibration]
\label{Rem-QuatHopfFib}
One exception occurs in degree $8$, where we have a $k$-invariant 
coming from the quaternionic Hopf fibration, and takes 
the form 
$$
\xymatrix{
k:(S^4)_2\ar[r] & K(\Z_{12},8)\times K(\Z,8)\;.
}
$$
Mapping out of $Y^{11}$, we identify the projection to the factor $K(\Z,8)$ as $\phi_2^*(\widetilde{G}_4^2)$, 
where $\phi_2 :(S^4)_2\to K(\Z,4)$ is the map at the second stage. Killing the $k$-invariant at this stage 
corresponds to a choice of trivialization $\delta \widetilde{C}_7=\widetilde{G}_4^2$. From a complimentary point of view, this case is discussed in detail
in \cite{FSS-twist}\cite{FSS-level}. 
\end{remark} 

The following statement follows directly from the identification of the homotopy groups 
of $S^4$ in the relevant degrees and assembling them into the tower
one degree at a time. 
See \cite{To} for a tabulation corresponding to the stages.

\begin{lemma}[Unstable Postnikov tower of $S^4$] \label{unstpost}
Overall, the Postnikov tower takes the following form

$$
\xymatrix@C=2em@R=1.5em{
 & S^4\ar@{}[d]^{\!\!\!\!\! \vdots} 
\\
K(\Z_{15},11)\ar[r] & (S^4)_7\ar[d]
\\
K(\Z_{24}\times \Z_3,10)\ar[r] & (S^4)_6\ar[d]\ar[r] & K(\Z_{15},12)
\\
K(\Z_2 \times \Z_2,9)\ar[r] & (S^4)_5\ar[d]\ar[r] &K(\Z_{24}\times \Z_3,11)
\\
K(\Z_2 \times \Z_2,8)\ar[r] & (S^4)_4\ar[d]\ar[r] & K(\Z_2 \times \Z_2,10)
\\
  K(\Z_{12},7)\times K(\Z,7)\ar[r] &(S^4)_3 \ar[d] \ar[r] & K(\Z_2 \times \Z_2,9)
\\
   K(\Z_2,6)\ar[r] & (S^4)_2 \ar[d]\ar[r]^-{(-,\; \iota_4^2)}  &   K(\Z_{12},8)\times K(\Z,8)
\\
  K(\Z_2,5)\ar[r] & (S^4)_1 \ar[d]\ar[r]^-{\alpha_7}  &  K(\Z_2,7)
   \\
     & (S^4)_0=K(\Z,4) \ar[r]^-{{\rm Sq}^2\iota_4} & K(\Z/2,6)
}
$$
where we have identified the first few obstructions. \end{lemma} 
%\proof
%This follows directly from the identification of the homotopy groups 
%of $S^4$ in the relevant degrees and assembling them into the tower
%one degree at a time. 
%See \cite{To} for a tabulation.
%\endproof

The unlabeled obstructions in Lemma \ref{unstpost} are unknown
 to us and seem to be quite complicated. We hope to revisit this in the future.

\newpage 

%%%%%%%%%%%%
\subsection{Physical manifestations and examples} 
%Torsion in cohomotopy} 
\label{Sec-Ex} 
%%%%%%%%%%%%%%%

Most of the obstructions for lifting cohomology classes to cohomotopy are torsion obstructions, as 
we have seen.
Given that the fields $G_4$ and $G_7$ are classes which appear in cohomology with real coefficients, 
it is natural to wonder how torsion obstructions could impose constraints on these classes. For instance, 
we saw that there is an  obstruction $\alpha_7$ which acts on $\Z_2$-classes in degree 5. At first glance, 
this might seem awkward since no fields of degree 5 in M-theory -- see also Remark \ref{Rem-pure5}. 
In this section we offer further physical interpretations of this and similar obstructions.

\medskip
Many of the  anomaly cancellation conditions present in the 
M-theory literature require an integral lift of a real cohomology class.

\begin{remark}[The anomaly in the partition function]
 Quantization in cohomotopy yields the condition ${\rm Sq}^2(\widetilde{G}_4)=0$ for 
some integral lift of $G_4$. As highlighted in \cite{FSS-twist}, this immediately implies 
the vanishing of the DMW anomaly \cite{DMW} ${\rm Sq}^3(\widetilde{G}_4)=0$. 
From the obstruction theory for $S^4$, we have an exact sequence of pointed sets
(cf. relation \eqref{deg5})
\(\label{exseqdc}
\hspace{-5mm} 
\xymatrix{
0\ar[r]&  H^5(Y^{11};\Z_2)/{\rm Sq}^2H^3(Y^{11};\Z)\ar[r]^-{i_*} & 
[Y^{11},(S^4)_1]\ar[r]^-{p_*} & \big\{x\in H^4(Y^{11};\Z);  {\rm Sq}^2(x)=0\big\}\ar[r] & 0 
},
\)
where $i_*$ is induced by post-composition with the fiber inclusion $K(\Z_5,5)\to (S^4)_1$ and $p_*$ is induced by post-composition with the fibration $(S^4)_1\to K(\Z,4)$.
In \cite{DMW} it was shown that the phase of the partition function for the $C$-field on 
$X^{10}\times S^1$ is $\pm 1$, depending on the the vanishing of a function 
$f:H^4(X^{10};\Z)\to \Z_2$. Now $f$ is not linear, but obeys the relation 
\(\label{dmwsq2p}
f(a+b)=f(a)+f(b)+\int_{X}a\cup {\rm Sq}^2(b)
\)
and $f(a)=0$ when $a=0$. In their notation,
% of \cite{DMW}, 
$a$ is a choice of integral lift of the $G_4$. 
The discussion in \cite[Section 6.2]{DMW} notes that if $f$ were linear, then the contribution of $a$ 
to the partition function should vanish unless $f(a+c)=f(a)$ with $c$ torsion (i.e., $f$ should not actually 
depend on the choice of integral lift). However, the last term on the right of \eqref{dmwsq2p} prevent $f$ 
from being linear. To circumvent this issue, the authors consider the subset $L^{\prime}$ of all 
torsion $c\in H^4(Y^{11};\Z)$ such that \footnote{Actually the weaker condition 
${\rm Sq}^3(c)=0$ is considered, but the discussion works
 equally well if we pass to this smaller class.} ${\rm Sq}^2(c)=0$ and analyze the nonvanishing 
 conditions of the phase
\(\label{parphsdmw}
\sum_{c\in\,  L^{\prime}}(-1)^{f(a+c)}.
\)
What is interesting is that the condition that $c$ lift to cohomotopy already forces $c$ to be in 
$L^{\prime}$, by the exact sequence \eqref{exseqdc}. In fact, the calculation using the torsion 
pairing in \cite[p. 42]{DMW} also shows that \footnote{In \cite{DMW}, ${\rm Sq}^3$ is used, 
but the same discussion works with 
${\rm Sq}^2$ by letting $M$ be $H^6(X^{10};\Z_2)$ and using the cup product pairing 
$$\int_{X^{10}}(-)\cup (-):L^{\prime}\times M\longrightarrow \Z_2$$
directly instead of the induced torsion pairing.} 
the condition on $a$ becomes that (after possible 
modification by a torsion class) ${\rm Sq}^2(a)=0$.  It follows that the fields which contribute 
to the phase \eqref{parphsdmw} are just the field which lift to the first Postnikov stage in cohomotopy. 
\end{remark} 

\begin{remark}[Mod 2 invariant and geometric submanifolds] 
\label{rem1-PT}
There is another mod 2 invariant which can be defined using cohomotopy. Recall that by 
Pontrjagin-Thom theory, $[Y^{11},S^4]$ can be identified with framed bordism classes of 
7-dimensional submanifolds. Let $M$ be a 7-dimensional submanifold defined by a map 
$Y^{11}\to S^4$ and let $\phi:M\times \R^4\to \mathcal{N}$ be the framing of the normal bundle. 
Then a choice of volume form $\omega$ on $Y^{11}$ naturally gives rise to a volume form 
$\omega_{\phi}$ on $M$ by contracting out the four unit normal vector fields, defined via $\phi$. 
Moreover, if $\omega$ is integral on $Y^{11}$, then so is $\omega_{\phi}$. This gives an assignment
$$
\xymatrix{ 
[Y^{11},S^4]=\{([M],\phi), M\subset Y^{11}\}\ar[r]& \int_{M}\omega_{\phi} \mod 2\in \Z_2
}.
$$
This assignment is additive with respect to disjoint union, defines a group homomorphism,
and gives the parity of the volume of $M$. We will come back to this in Remark \ref{rem2-PT}.
\end{remark} 

\begin{remark}[Lifts of integral cohomology classes to K(O)-theory]
 As we saw above, in order to read the condition (see also \cite{FSS-twist})
\(
\label{sq2g4}
{\rm Sq}^2(G_4)=0
\)
properly, one needs to choose an integral lift $\widetilde{G}_4$ of $G_4$ and there is no canonical way 
to do this. For the analogous case of ${\rm Sq}^2(F_4)$, where $F_4$ is the Ramond-Ramond (RR) 
field from which $G_4$ is lifted from $X^{10}$ to $Y^{11}=X^{10} \times S^1$, this is interpreted 
as a condition on an integral lift of $F_4$ in order that it lift to K-theory \cite{DMW} (see 
\cite{GS5}\cite{GS-RR} for extensive discussions of such lifts). 
This indicates that the partition function of the RR fields is sensitive to the choice of integral lift 
of $F_4$ (in addition to other degrees as well). 
The condition at hand \eqref{sq2g4} provides an analogous sensitivity to the integral lift
$\widetilde{G}_4$ of $G_4$ as well as to lifting to KO-theory instead of K-theory.  
\end{remark}

%%%%%%%~%~%%%
%\section{Examples} 
%%%%%%%%%%%%%%%

\begin{remark}[Purely cohomotopic contribution] 
\label{Rem-purelyco}
We give an instance where cohomotopy gives a contribution even when the corresponding 
cohomology is trivial (complimenting Remark \ref{Rem-pure5}).  
The choice of generator of $H^4(S^4; \Z)$ defines a map $S^4 \to K(\Z, 4)$ and hence a homotopy
fibration sequence
$$
K(\Z, 3) \longrightarrow F \longrightarrow S^4 \longrightarrow K(\Z, 4)\;,
$$
with $F$ the homotopy fiber.
This gives an exact sequence of pointed sets 
$$
H^3(Y^{11}; \Z) \longrightarrow [Y^{11}, F] \longrightarrow
 [Y^{11}, S^4] \longrightarrow H^4(Y^{11}; \Z)\;.
$$
If $H^4(Y^{11}; \Z)=0=H^3(Y^{11}; \Z)$ then we get a bijection $[Y^{11}, F]=[Y^{11}, S^4]$.
We know that, by definition, $\pi_i(F)=0$ for $i\leq 4$, while $\pi_5(F)\cong \pi_5(S^4) \cong \Z_2$.
Hence, by cellular approximation, we get 
$$
[Y^{11}, F]\cong [Y^{11}, K(\Z_2, 5)]\cong H^5(Y^{11}; \Z_2)\;.
$$
Therefore, we get that degree 4 cohomotopy gives a contribution 
to cohomology in higher degree, in this case degree  five, 
$\big\vert[Y^{11}, S^4] \big\vert= \big\vert H^5(Y^{11}; \Z_2) \big\vert$. 
See also Remark \ref{Rem-pure5} for an interpretation. 
%\cite{LT}

%Is this a Stiefel-Whitney class? Use Karoubi-Weibel and our KO to detect. 
%If so then this is related to $W_5$ in \cite{DFM}.  
%{\color{red} This example is essentially the same as the one in remark 3.2. The degree 5 class is related 
%to ${\rm Sq}^2(G_4)$ be transgression. I'm sure you could pull out some relationships with Stiefel-Whitney 
%using Wu, but I don't think this is much related to $W_5$ (which is an integral class anyways). 
%When $W_5$ is zero, $w_4$ agrees with $\frac{1}{2}\lambda$ in cohomology with $U(1)$-coefficients 
%(this is  the discussion in \cite{DFM}), so I think $W_5$ should be more related to the twist. 
%Schematically,
%$$W_5=0 \Rightarrow j_2(w_4)={\rm exp}(\tfrac{1}{2}\lambda)\Rightarrow G_4=
%a-\tfrac{1}{2}\lambda ~~\text{not integral unless} ~~ w_4=0$$}
\end{remark}

%
%\paragraph{Group structure.} A space $X$ is called $n$-coconnected if the cohomology 
%group $H^q(X; \Z)=0$ for every $ q\geq n$. 
% In particular, $X$ is $n$-coconnected if $\mathrm{dim}X\leq n$.
% 
%\medskip
% If $X$ is $7$-coconnected cellular space, then $\pi^4(X)$ 
%forms an abelian group with the addition $\alpha+ \beta$ as its group operation
%and $0$ as its group-theoretic natural element. If $\alpha \in \pi^4(X)$ is represented
%by a map $\phi: X \to (S^4, s_0)$, then its negative, $-\alpha$, is represented 
%by the composition $r \phi$, where $r: (S^4, s_0)\to (S^4, s_0)$
%is an arbitrary map of degree $-1$. 
%{\color{red} I don't understand how we get a group structure here. Is that really the right definition for coconnected? Also, is this leading to examples?}
%{\color{green} The above bit is recycled from other files and was waiting to be looked at,
%and it appears in the proof of some of the examples. 
%Now that we see where we're going and we are not proving the examples, it is not relevant anyway.} 

\begin{examples}[Flux compactification spaces] \label{Flux-ex}
We consider the following examples, involving Anti-de Sitter space
${\rm AdS}_n$. This space is homotopically essentially trivial aside from the fundamental 
group. In order stay away from matters related to insisting the action of the fundamental 
group to be nice (e.g., nilpotent), we will assume simply-connectedness, which will 
ensure the homotopy techniques can be safely used. This then can be arranged by 
 taking the universal cover $\widetilde{\rm AdS}_n$ of ${\rm AdS}_n$. 
%This also follows from general arguments on the dimension. 

\begin{enumerate}[{\bf (i)}] 
\vspace{-2mm} 
\item \underline{$\widetilde{\rm AdS}_7 \times \R P^4$}: This example is important in considering M-theory on 
an orientifold \cite{Wi}\cite{Ho}.
The internal space $\R P^4$ is obtained by attaching a 4-cell 
to $\R P^3$ by the quotient projection $f_3: S^3 \to \R P^3$ which identifies the antipodal points. Collapsing 
the subspace $\R P^3 \subset \R P^4$ to a point yields a map $q_4: \R P^4 \to S^4$. This gives 
rise to an element $[q_4] \in \pi^4(\R P^4)$. Then, from \cite{We}, we have $\pi^4 (\R P^4)\cong\Z_2$ 
with generator $[q_4]$. Comparing with integral cohomology, $H^4(\R P^4; \Z)= 0$, indeed
shows that cohomotopy detects more.

\vspace{-2mm} 
\item \underline{$\widetilde{\rm AdS}_4 \times \C P^2 \times T^2$}: This example is important in 
supersymmetry without supersymmetry 
\cite{DLP} and T-duality \cite{BEM}. 
We will again take the covering space of the AdS factor. Furthermore, note that 
the $T^2$ factor does not contribute to cohomotopy due to dimension reasons. 
The complex projective space
$\C P^2$ is obtained by attaching a 4-cell to $\C P^1=S^2$ 
by the Hopf map $f_1: S^3 \to \C P^1$, which is also the Hopf map $\eta_2$ above. 
Collapsing $\C P^1=S^2 \subset \C P^2$ to a point yields a map $q_2: \C P^2 \to S^4$. 
Then, from \cite{We},  $\pi^4(\C P^2) \cong \Z$ with generator $[q_2]$. 
Comparing to cohomology, we have $H^4(\C P^2; \Z) \cong \Z$, so that in this case, the two 
coincide, so that no new information is supplied by cohomotopy. 

\vspace{-2mm} 
\item \underline{$\widetilde{\rm AdS}_7 \times \C P^2$}: The example is similar to the previous. 
Passing again to the simply connected cover of ${\rm AdS}_7$, the only nontrivial 
contribution again comes from $\pi^4(\C P^2)\cong \Z$, again with no extra contribution.  
%{\color{red} Added a bit here. Feel free to fill it out.}

\vspace{-2mm} 
\item \underline{$\widetilde{\rm AdS}_4 \times \R P^5 \times T^2$}: It follows from  \cite{We} that 
$\pi^4(\R P^5)$ is cyclic or order 4, i.e. either $\Z_4$ or $\Z_2 \times \Z_2$, with 
generator $[\eta_4 q_5]$ where $\eta_4: S^5 \to S^4$ 
is the 2-fold iteration of the Hopf map $\eta_2: S^3 \to S^2$, and $q_5$ is defined analogously to 
$q_4$ from above.  On the other hand, $H^4(\R P^5; \Z)\cong \Z_2$, so that 
there is further contribution from cohomotopy, either as an extra $\Z_2$ or as a $\Z_4$ vs. $\Z_2$.

%\attn{Witten's baryons} 

%\item ${\rm AdS}_4 \times \R P^6 \times S^1$: 
%The order is \cite{We} $|\pi^4(\R P^6)|=8$. 
%

\vspace{-2mm} 
\item \underline{$\widetilde{\rm AdS}_4 \times \C P^3 \times S^1$}: This example is also 
important in the phenomenon of 
supersymmetry without supersymmetry. Let $i_2: \C P^2 \hookrightarrow \C P^3$ denote 
the inclusion and $2: S^4 \to S^4$ a map of degree 2. Then, again invoking \cite{We},
$$
\pi^4(\C P^3) \cong \Z \oplus \Z_2
$$
where the generator of $\Z$ is $[\alpha_3]$ where $\alpha_3 i_2 \simeq 2q_2$ and the 
generator of $\Z_2$ is $[\eta_4 \eta_3 q_3]$. Comparing to cohomology, 
we have $H^4(\C P^3; \Z) \cong \Z$, so that there is an extra contribution of $\Z_2$ present in 
cohomotopy. 
\end{enumerate} 
\vspace{-3mm} 
We have seen that in several backgrounds there is an extra torsion contribution from cohomotopy
over integral cohomology. This is an interesting effect that deserves further investigation, to which
we hope to get back elsewhere.
\end{examples} 

 \medskip
\begin{examples}[Quaternionic and octonionic projective planes] 
\label{Ex-fluxcomp2}
Similarly for $\HH P^2$ 
and $\mathbb{O}P^2$, we have the following, again making use of some of the constructions in \cite{We}.
\vspace{-1mm} 
\item {\bf (i)}  \underline{For $\HH P^2$}: Consider the Puppe sequence or the mapping 
cone sequence of the quaternionic Hopf fibration
$$
\xymatrix{
S^{7} \ar[r]^{h_\HH} & S^4 \ar[r]^p & \HH P^2 \ar[r]^q & S^{8}
 \ar[r]^{\Sigma h_{\HH}} & S^5 \ar[r] & \hdots 
}
$$
Now apply the 2-fold suspension $\Sigma^{2}$. This gives
$$
\xymatrix{
S^{9} \ar[rr]^{\Sigma^{2}h_\HH} && S^{6} \ar[r]^-{\Sigma^{2} p} 
& \Sigma^{2} \HH P^2 \ar[r]^-{\Sigma^{2} q} &
S^{10} \ar[rr]^{\Sigma^{2} h_{\HH}} && S^{7} \ar[r] & \hdots 
}
$$
Taking the cohomotopy groups gives the long exact sequence
$$
\xymatrix@C=1.5em{
\pi^{6}(S^{9}) \ar[rr]^-{(\Sigma^{2}h_\HH)^*} && \pi^{6}(S^{6}) \ar[rr]^-{(\Sigma^{2} p)^*} && 
\pi^{6}(\Sigma^{2} \HH P^2) \ar[rr]^-{(\Sigma^{2} q)^*} &&
\pi^{6}(S^{10}) \ar[rr]^-{(\Sigma^{2} h_{\HH})^*} && \pi^{6}(S^{7}) \ar[r] & \hdots 
}
$$
Now 
$\pi^{6}(S^{9}) \cong \pi_{9}(S^{6}) \cong \Z_{24}$,
$\pi^{6}(S^{6})\cong \pi_{6}(S^{6})\cong \Z$, and 
$\pi^{6}(S^{10})=0$, so we have a sequence 
$
\Z_{24} \to \Z \to  A \to  0\
$.
As a group homomorphism 
from $\Z_{24}$ to $\Z$ necessarily vanishes (since the image consists of zero divisors), 
this  gives $\pi^{6}(\Sigma^{2} \HH P^2)\cong A \cong \Z$. 
The exact sequence 
above gives the isomorphism.
Therefore, 
\(
\pi^4(\HH P^2)\cong \Z\;.
\label{pi4HP2}
\)
As in the complex case, this agrees with cohomology, $H^4(\mathbb{H} P^2; \Z)\cong \Z$, and hence 
no new contribution, but of course there is a compatibility. The quaternionic projective plane is also 
a compatification space for M-theory and also appears in anomaly cancellation (see \cite{FSS-twist}).

\vspace{0mm} 
\item {\bf (ii)} 
%Similar sequences exist for $\C P^2$ and the 
%For the Cayley plane $\mathbb{O} P^2$. 
%Since we have already pointed out that $\pi^4(\C P^2)\cong \Z$ in example 2 and 3 above, we illustrate with 
\underline{For $\mathbb{O} P^2$}:  In the octonionic case we have a cofiber sequences
$
S^{15}\longrightarrow S^8\longrightarrow \mathbb{O} P^2\longrightarrow S^{16}\longrightarrow S^{9},
$
which (after suspending 4-times) yields
$$
\pi_{19}(S^8)\longrightarrow \pi_{12}(S^8)\longrightarrow \pi^8(\Sigma^4\mathbb{O} P^2)\longrightarrow \pi_{20}(S^8)\;.
$$
Identifying low-dimensional homotopy groups of spheres gives the exact sequence
$
\Z_{1008}\to 0 \to \pi^4(\mathbb{O} P^2)\to 0
$, so that 
$$
\pi^4(\mathbb{O} P^2)\cong 0\;.
$$
%{\color{blue} Below is no longer relevant, since cohomology agrees with cohomotopy. Can you fix this?}
This is similar to the complex and quaternionic cases, although the comparison to 
to cohomology is different, in that we also have $H^4(\mathbb{O}P^2; \Z) =0$.
Perhaps this is not surprising, as the dimension takes
us outside those of critical M-theory and string theory, but are 
very interesting for the bosonic case (see \cite{OP2}\cite{supermultiplet}). 

\medskip
The effects in these examples of projective spaces also deserve further 
investigation. 
\end{examples}

%%%%%%%%%%%%
\section{Differential refinements:  $\mathbf{B}^3U(1)_\nabla$ vs. $\widehat{S^{\,4}}$} 
\label{Sec-DiffCoh} 
%%%%%%%%%%%%%

%%%%%%%%%%%%%%%%%%%%%
\subsection{Differential cohomotopy} 
\label{Sec-diffcoh} 
%%%%%%%%%%%%%%%%%%%%

Here we expand on the discussion of differential cohomotopy in \cite{FSS2}.
As with any differential refinement, differential cohomotopy involves an interplay between 
topological information on smooth manifolds and the geometric information of differential 
forms via a general de Rham type theorem. The basic ingredients for this 
general machinery can be found in \cite{Cech}\cite{SSS3}\cite{FSS14c}
and our discussion here will assume familiarity with these 
ingredients. We encourage the reader to consult these references for more details as needed.

\medskip 
We recall the smooth category of cartesian spaces, which we denote $\cartsp$. The objects in 
this category are convex open subsets $U\subset \R^n$, and the morphisms are smooth 
maps $f:U\to V$. This category admits a Grothendieck topology generated by good open covers.
 By {\it smooth stack}, we mean an $\infty$-groupoid valued functor on the site of cartesian 
 spaces that satisfies descent. As a concrete model, we will work with simplicial presheaves on 
 $\cartsp$, equipped with the local projective model structure (where local means Bousfield 
 localized at {\v C}ech nerves of good open covers). We denote this model category as 
$$
\sPSh(\cartsp)_{\rm loc}:={\rm Fun}(\cartsp^{\rm op},\sset)_{\rm proj,loc}\;.
$$

\begin{remark}[Simplicial presheaves] 
\label{constsh}
Most of our categorical constructions will take place in simplicial presheaves. In particular, we can regard 
a simplicial set as a simplicial presheaf via the constant stack functor 
  $$
  \delta:\sset\to \sPSh(\cartsp)_{\rm loc}, \qquad \delta(X)(U):=X\;.
  $$
  We can also regard an ordinary $\set$-valued presheaf $F:\cartsp^{\rm op}\to \set$ as a simplicial 
  presheaf via the left Kan extension 
  $$
  i_!:\PSh(\cartsp)\to \sPSh(\cartsp)_{\rm loc}
  $$ 
  along the inclusion functor $i:\cartsp\into \cartsp\times \Delta$ that sends $U\mapsto (U,[0])$. 
  Concretely, $i_!(F)(U)$ is the nerve of the discrete groupoid with object $F(U)$ and only identities 
  as morphisms. Whenever an ordinary presheaf appears with simplicial presheaves, we implicitly 
  embed the presheaf using $i_!$ and we will not include $i_!$ in the notation.
\end{remark}

\medskip
Let $\mathfrak{s}^4$ be the Lie 7-algebra whose corresponding Chevellay-Eilenberg algebra is the exterior 
algebra on generators $g_4$ and $g_7$ with relations (see \cite[Sec. 3]{FSS2} for details) 
$$
dg_4=0 ~,  \qquad dg_7=g_4\wedge g_4\;.
$$
As a de Rham model for flat 1-forms with values in $\mathfrak{s}^4$, 
we take the sheaf on the site of cartesian spaces 
given by the assignment 
$$
\xymatrix{
\Omega_{\rm fl}^1(-;\mathfrak{s}^4): U
\ar@{|->}[r] &
\hom_{\rm dgcAlg}({\rm CE}(\mathfrak{s}^4),\Omega^*(U))
},
$$
for each cartesian space $U\cong \R^n$. Here the morphisms in the set on the right are taken in differentially 
graded commutative algebras. Following Remark \ref{constsh}, we can view this sheaf as a zero-truncated 
smooth stack (which we denote by the same symbol). 

\medskip
Every smooth stack has a homotopy type associated to it, given by applying the left adjoint 
to the constant functor $\delta$. The existence of this functor is part of the axioms of a cohesive 
$\infty$-topos (see \cite[Section 4.1]{Urs}). Concretely, if we are given a smooth stack $X$, 
the homotopy type is given by the formula \cite[Corollary 6.4.28]{Urs}
\begin{equation}\label{realize}
\vert X \vert\simeq \hocolim_{[n]\in \Delta^{\rm op}}X(\Delta^n).
\end{equation}
By the above formula, we see that the homotopy type of $\Omega_{\rm fl}^1(-;\mathfrak{s}^4)$ can be computed 
via the Sullivan construction\footnote{This construction is essentially the same as the familiar construction of 
a rational space in rational homotopy theory, but over the field $k=\R$. Note however, that we have taken smooth 
forms instead of polynomial forms. That this agrees with the usual Sullivan construction follows readily from 
the fact that $A^*_{\rm PL}(\Delta^n)\into \Omega^*(\Delta^n)$ is a quasi-isomorphism of complexes. 
See \cite{FSS2}.} 
as the $\R$-local 4-sphere, which we denote $S^4_{\R}$. 
Then, taking the homotopy pullback along the localization map 
$L_{\R}:S^4\to S^4_{\R}$ and the unit of the adjunction $\textstyle \int:\Omega^1_{\rm fl}(-;\mathfrak{s}^4)\to \delta(\vert \Omega_{\rm fl}^1(-;\mathfrak{s}^4)\vert)$, we get a smooth stack
$$
\xymatrix@R=1.5em@C=3.5em{
\widehat{S}^{\,4}\ar[r]\ar[d] & \Omega_{\rm fl}^1(-;\mathfrak{s}^4)\ar[d]^-{\textstyle{\int}}
\\
\delta(S^4)\ar[r]_-{\delta(L_{\R})} & \delta(S^4_{\R}) \;.
}
$$
We have the following natural notion. 
\begin{defn}[Differential unstable cohomotopy]\label{unstablecoh}
For a smooth manifold $X$, the differential cohomotopy of $X$ in degree 4 is defined as the pointed set 
$$
\widehat{\pi}_u^{\,4}(X):=\pi_0\text{Map}\big(X,\widehat{S}^{\,4}\big),
$$
where the mapping space on the right is the derived mapping space between smooth stacks and $X$ 
is viewed as a sheaf on $\cartsp$ by the assignment $U\mapsto C^{\infty}(U,X)$ (see Remark \ref{constsh}).
\end{defn}
This gives a geometric model for \textit{unstable cohomotopy}, but we will also need 
a geometric model for \textit{stable cohomotopy}. In \cite[Section 4.4]{BNV}, following the construction 
in \cite[Section 4]{HS}, it was shown that there is a sheaf of spectra modelling a differential cohomology 
theory (in the sense of \cite{SSu}) refining a given underlying cohomology theory $E^*$. In our case, we 
are concerned about the cohomology theory given by stabilizing cohomotopy in degree 4. Stably, $S^4$ 
has only torsion groups in 
higher degrees and hence the canonical map $S^4\to K(\R,4)$ is a stable $\R$-local equivalence. 
Geometrically, the realification is modeled by closed differential 4-forms $\Omega^4_{\rm cl}(-)$. Hence, according to \cite[Section 4.4]{BNV}, the differential refinement of cohomotopy in degree 4 can be taken to be the homotopy pullback
$$
\xymatrix@R=1.5em@C=3em{
\widehat{\Sigma^{\infty}S^{\, 4}}\ar[r]\ar[d] & H\Big(\tau^{\leq 0}\Omega^{4+*}(-)\Big)\ar[d]
\\
\delta(\Sigma^{\infty}S^4)\ar[r] & \delta(\Sigma^4H\R)\;,
}
$$
where $\Omega^{4+*}(-)$ denotes the de Rham complex, shifted so that $\Omega^4$ is in 
degree zero, and $\tau^{\leq 0}$ truncates the complex in degree zero so that the complex is 
concentrated in negative degrees. The functor $H$ denotes the Eilenberg-MacLane functor (see e.g. \cite{Shi}) which turns a chain complex into a spectrum and $\Sigma^{\infty}$ denotes the infinite suspension functor, which associates a spectrum to a space (or simplicial set). 

\begin{defn}[Differential stable cohomotopy]
Let $X$ be a smooth manifold. The \textit{stable} differential 
cohomotopy group of $X$ is defined as
$$
\widehat{\pi}_s^{\,4}(X):=\pi_0\text{Map}\big(X,(\widehat{\Sigma^{\infty}S^{\,4}})_0\big),
$$
where the subscript $0$ denotes the degree zero stack of the sheaf of spectra $\widehat{\Sigma^{\infty}S^4}$ (i.e., its infinite loop stack). The manifold $X$ is viewed as a smooth stack as in Definition \ref{unstablecoh}.
\end{defn}

Ultimately, we will be most interested in the above unstable version of differential cohomotopy. 
However, the stable version will be useful as an approximation and is topologically easier to analyze 
(as we have seen in \cref{Sec-EMvsS4}).

\medskip
\noindent
{\bf Geometric meaning of cocycles.} We now discuss a geometric interpretation for cocycles in 
differential cohomotopy. More precisely, we address what type of geometric data a differential 
cocycle $\hat{c}:M\to \widehat{S}^{\,4}$ classifies.

\newpage 

\begin{defn}[Geometric cohomotopy cocycles]
If $X$ is a smooth manifold, a morphism $\hat{c}:X\to \widehat{S}^{\, 4}$ can be identified with a 
triple $(c,h,\omega)$ where 
\begin{itemize}
\vspace{-2mm} 
\item
$c:X\to S^4$ is a cocycle in ordinary cohomotopy, 
\vspace{-2mm} 
\item
$\omega:{\rm CE}(\mathfrak{s}^4)\to \Omega^*(X)$ is a DGA morphism, determined by specifying forms 
$\omega_4$ and $\omega_7$ on $M$ satisfying $d\omega_7=\omega_4^2$ and $d\omega_4=0$,
\vspace{-2mm} 
\item
 and $h$ is a homotopy interpolating between the rational cocycle represented by the form data and the 
 rationalization of the classifying map $c:X\to S^4$. Thus, $h$ exhibits a sort of de Rham theorem for 
 cohomotopy. 
\end{itemize} 
\end{defn} 

\medskip
\begin{remark}[Relation to the Pontrjagin-Thom (PT) construction] 
\label{rem2-PT}
\vspace{-2mm} 
%\item {\bf (i)} 
Recall from Remark  \ref{rem1-PT} that by the PT construction, a mapping $c:X\to S^4$ classifies 
a bordism class of framed 
codimension 4 submanifolds of $X$. This correspondence realizes the codimension 4 submanifold 
$M$ as the preimage of a fixed regular value on $S^4$ and maps the closure of a tubular neighborhood 
of $M$ in $X$ onto $S^4$ via the given framing of the normal bundle  
$\mathcal{N}\cong \R^4\times M\to S^4\times M\overset{\rm pr}{\to} S^4$. 
%\vspace{-1mm} 
%\item {\bf (ii)} 
Hence, the cocycle $\hat{c}$
gives in particular the data of a codimension 4 submanifold $M\subset X$. It also gives a choice 
of fiberwise volume form $\omega_4=c^*g_4$ of the trivial sphere bundle, where 
$g_4\in {\rm CE}(\mathfrak{s}^4)$ is identified with a choice of volume form for the sphere $S^4$. 
%\vspace{-1mm} 
%\item {\bf (iii)} 
Much more could be said about the geometric model provided by the Pontrjagin-Thom 
equivalence, but this falls outside the scope of the present paper. We only include this brief discussion 
to provide some conceptual geometric intuition.

\noindent In view of geometric interpretation via volume forms, we can introduce dynamics by 
throwing in a  radius as a parameter, viewed as the breathing mode (see, e.g., \cite{LS}). 

\end{remark}

%\begin{remark}
%Note that $\widehat{Sq}^3(\hat{x})$ may be nonvanishing even when $Sq^3(\rho_2(x))$ vanishes. In this case, however, $\widehat{Sq}^3(\hat{x})$ can be identified with the image of a differential form $C$, under the map 
%$$a:\Omega^{p+2}(X)\to \widehat{H}^{p+3}(X;\Z)\;.$$
%We will come back to this point in a moment.
%\end{remark}
%
%To build the Postnokiv tower in the differential setting, recall that we have a fiber sequence
%$$\Omega^{\leq p}\to \mathbf{B}^{p}U(1)_{\nabla}\to K(\Z,p+1)\;,.$$
%Then if $k_p$ is the $k$-invariant at the $p$-th stage of the Postnikov tower, we have 
%$$
%\xymatrix{
%\widehat{X}_{p+1}\ar[r] & \widehat{X}_p\to X_p\ar[d]
%\\
% & \mathbb{B}^{p-1}U(1)_{\nabla} & K(\Z;p)
%}
%$$

%%%%%%%%%%%%%%%%%%%%%%%%%%%%%%%%%%
\subsection{Torsion obstructions in differential cohomology}
\label{Sec-TorDiff} 
%%%%%%%%%%%%%%%%%%%%%%%%%%%%%%%%%%%%%%%%

We saw in \cref{Sec-Z} that the Postnikov tower for the 4-sphere has many $k$-invariants which are 
torsion classes. For our physics applications, the tower must be refined to obtain an obstruction theory 
for lifting cohomotopy classes to the differential refinement of cohomotopy and it is not completely 
clear how to deal with such obstructions in the refinement. Indeed, the obstruction theory for differential 
refinements is obtained by Chern-Weil form representatives of the $k$-invariants, and one requires these 
forms to trivialize when the topological obstructions vanish (the choice lift through the next stage in the 
tower gives rise to the trivialization). Since Chern-Weil theory is not available for torsion classes, we need 
to find an alternative method for the differential refinement. 

\medskip
Recall that the moduli stack of circle $n$-bundles with connection fits into a homotopy pullback diagram 
\cite{Cech} \cite{SSS3}\cite{FSS14c}
$$
\xymatrix@C=3em@R=1.5em{
\mathbf{B}^nU(1)_{\nabla}\ar[r]^-{R}\ar[d]_-{I} 
& \Omega^{n+1\ar[d]}_{\rm cl}
\\
K(\Z,n+1)\ar[r] & \Omega^{\leq n+1}_{\rm cl}
}
$$
where $\Omega^{\leq n+1}_{\rm cl}$ is obtained by applying the Dold-Kan functor to the sheaf of 
positively graded chain complexes
$$
\Omega^{\leq n+1}_{\rm cl}:=\Gamma\big( \hdots \longrightarrow 0 \longrightarrow \Omega^0\longrightarrow 
\hdots \longrightarrow \Omega^{n+1}_{\rm cl}\big)
$$
and $K(\Z,n+1)\to \Omega^{\leq n+1}_{\rm cl}$ is induced by the inclusion $\Z\into \Omega^0$. 
\footnote{The smooth stack $\Omega^{\leq n+1}_{\rm cl}$ represents cohomology with $\R$ 
coefficients in degree $n+1$. In fact, the canonical inclusion 
$$
\xymatrix@R=1em{
\R\ar[d]\ar[d] \ar[r] & 0  \ar[r] \ar[d] & \cdots \ar[r] \ar[d] & 0 \ar[d] 
\\
\Omega^0\ar[r]^-{d} & \Omega^1\ar[r]^-{d} &
\cdots\ar[r] 
& \Omega^{n+1}_{\rm cl}
}
$$
is a quasi-isomorphism of sheaves of complexes, inducing an isomorphism 
$H^{n+1}(X;\R)\cong H^0\big(X;\R[n+1]\big)\cong H^0\big(X;\Omega^{\leq n+1}_{\rm cl}\big)$.}

\begin{remark}[Integral lifts of differential forms]
Consider any map $\hat{k}:\mathbf{B}^nU(1)_{\nabla}\to K(\Z_p,m)$. 
Since $K(\Z_p,m)$ is a geometrically discrete (i.e., a constant stack), the map 
$\hat{k}$ factors through the corresponding topological realization of the domain 
as 
$$
\xymatrix{ 
\hat{k}:\mathbf{B}^nU(1)_{\nabla}\ar[r]^{I} &
K(\Z,n+1)\ar[r]^{k} & K(\Z_p,m)
}.
$$
By the pasting law for pullbacks, we have iterative fiber products
$$
\xymatrix@C=3em@R=1.5em{
\widehat{F}\ar[d]\ar[r]\ar[d] & F\ar[d]\ar[r] & \ast \ar[d]
\\
\mathbf{B}^nU(1)_{\nabla}\ar[d]\ar[r]^{I}& K(\Z,n+1)\ar[r]^-{k} \ar[d]& K(\Z_p,m)
\\
\Omega^{n+1}_{\rm cl}\ar[r] & \Omega^{\leq n+1}_{\rm cl}
}
$$
For each fixed manifold $X$ mapping to the diagram, this naturally gives rise to a map
$$
\xymatrix{
[X,\widehat{F}]\ar[r] & \Omega^{n+1}_{\rm cl}(X)\times_{H^{n+1}(X;\R)} [X,F]}.
$$
The group on the right can be identified with differential forms whose de Rham class is in the image of 
the composite 
$$
\xymatrix{
[M,F]\ar[r] & H^{n+1}(M;\Z)\ar[r] & H^{n+1}(M;\R)}
.
$$
Such conditions can be realized as conditions on the possible integral lifts of differential forms. 
\end{remark} 

\begin{remark}[From differential forms to torsion constraints]  
From the above discussion, we see that if we take the usual fiber at Postnikov stages with torsion 
$k$-invariants, then differential forms still detect this information. More precisely, passing to the 
fiber leads to more constrained quantization conditions on the differential forms. This is precisely 
what is needed for our applications and we will treat torsion obstructions in this manner. 
\end{remark} 

\begin{example}[Constraints associated with reduction of coefficients] 
Let us take $k$ to be the mod $p$ reduction $k=\rho_p:K(\Z,n+1)\to K(\Z_p,n+1)$.
Then $F$ is easily seen to be $K(\Z,n+1)$ and the canonical map out of the fiber is 
$$
\xymatrix{ \times p:K(\Z,n+1)\ar[r] & K(\Z,n+1)}.
$$ 
Hence, classes in $[M,\widehat{F}]$ give rise to closed forms which, when paired with cycles gives an integer 
divisible by $p$. 
Such divisibility conditions, in the context of describing fields via K(O)-theory,
 are discussed extensively in \cite{GS7}\cite{GS-RR}. 
\end{example}

\begin{example}[Obstructions via refinement of cohomology operations] 
Consider the refinement of the Steenrod square 
${\rm Sq}^2$, given by the composition \cite{GS2}
$$
\xymatrix{
\widetilde{\rm Sq}^3:\mathbf{B}^mU(1)_{\nabla}\ar[r]^-{I} &
 K(\Z, m)\ar[r]^-{\rho_2} &
  K(\Z_2,m)\ar[r]^-{{\rm Sq}^2} &K(\Z_2,m+2)
  }.
$$
This is almost, but not quite, the differential refinement of ${\rm Sq}^3$ discussed in \cite{GS2}. 
The two become the same after including $\widetilde{{\rm Sq}}^3$ into differential 
cohomology via the map 
$$
\xymatrix{
K(\Z_2, m+2)\ar[r]&
 K(U(1), m+2)\simeq \mathbf{B}^{m+2}U(1)_{\nabla\text{- flat}}
 \; \ar@{^{(}->}[r] & \mathbf{B}^{m+2}U(1)_{\nabla}
 }
$$
induced by the inclusion $\Z_2\into U(1)$ via the 2-roots of unity. Let 
$K:=\ker\big({\rm Sq}^2\rho_2: H^m(- ;\Z)\to H^{m+2}(- ;\Z_2)\big)$. 
Then classes in $[Y,\widehat{F}]$ give rise to forms admitting integral lifts 
which are in the image of $K\into H^m(Y;\Z)$. 
For the field $G_4$ in spacetime $Y$, we take $m=4$, so that 
differential cohomotopy classes $[Y,\widehat{F}]$ are given by 4-forms 
$G_4^{\rm form}$ admitting integral images in the image of 
$\ker\big({\rm Sq}^2\rho_2: H^4(M;\Z)\to H^{6}(M;\Z_2)\big)
\into H^4(Y;\Z)$.
\end{example}

%The shift in degree from $Sq^2$ to $\widehat{Sq}^3$ comes from the commutative triangle 
%$$
%\xymatrix{
%&\mathbb{B}^{p}U(1)_{\nabla}\ar[dr]^-{\mathcal{I}}&
%\\
%\delta(K(U(1),p))\ar[rr]_{\delta(\beta)}\ar[ru]^-{j} && \delta(K(\Z,p+1))
%}
%$$
%and the Adem relation relation $Sq^1Sq^2=\rho_2\beta Sq^2=Sq^3$. That is, 
%$\widehat{Sq}^3$ is a differential %cohomology class refining $Sq^3$ in the 
%sense that $\rho_2\mathcal{I}(\widehat{Sq}^3)=Sq^3$. 

%
%The previous considerations show that torsion $k$-invariants do provide genuine geometric 
%obstructions -- namely, the vanishing of such obstructions forces divisibility conditions on 
%differential forms. Motivated by these observations, we will generally deal with the differential 
%refinement of the torsion Postnikov tower as follows. 
%\begin{itemize}
%\item For non-torsion $k$-invariants, we consider refinements via Chern-Weil theory. 
%In this case, the next stage in the Postnikov tower is given by pullback of the refinement by the map 
%$$a:\Omega^{\leq n}\to \mathbf{B}^nU(1)_{\nabla}\;.$$
%\item For torsion $k$-invariants, we simply take the fiber of the map as usual
%\end{itemize}

%%%%%%%%%%%%%%%%%%%%%
\subsection{Differential cohomotopy vs. differential cohomology} 
\label{Sec-DiffCohvsDiffCoh}
%%%%%%%%%%%%%%%%%%%%

In this section, we refine the Postnikov tower (see Lemma \ref{Lem-IntPost}) to the setting of 
differential cohomology. Our strategy for building the Postnikov tower for $\widehat{S}^{\, 4}$ 
stems from the basic observation that we can split this construction into the following three 
more elementary constructions.
\begin{enumerate}[{\bf (i)}]
\vspace{-2mm} 
\item The Postnikov tower in the opposite category of  DGCA's.
\vspace{-3mm} 
\item The ordinary Postnikov tower in spaces.
\vspace{-3mm} 
\item The Postnikov tower in spaces localized at $\R$. 
\end{enumerate}
It turns out that the process of differential refinement is compatible (in a certain sense) with the 
Postnikov construction. Before proving that this is the case, we begin with a preliminary observation.

\begin{lemma}[Postnikov system in the Sullivan construction] 
\label{postsul}
Let $\Lambda V$ be a Sullivan algebra on a graded vector space $V$, such that $H^1(\Lambda V)=0$. 
Let $\Lambda V_{\leq k}$ denote the subcomplex whose elements are spanned by wedge products of 
elements of $V$ in degree $\leq k$. Let $X_k:={\cal K}( \Lambda V_{\leq k} )$ denote the Sullivan 
construction of this subcomplex. Then the sequence of maps
$$
\left\{\hdots \to X_k\to X_{k-1}\to \hdots \to X_0  \right\},
$$
induced by the inclusions $\Lambda V_{\leq k-1}\into \Lambda V_{\leq k}$, is a Postnikov system for $X$. 
\end{lemma} 
\begin{proof}
Let $V_k\subset V$ denote the subspace of elements in degree $k$. First, by \cite[Proposition 17.9]{FHT}, 
the inclusions $\Lambda V_{\leq k-1}\into \Lambda V_{k}$ are sent to fibrations $X_k\to X_{k-1}$, with 
fiber $F_k:={\cal K}( \Lambda V_k)$, by the Sullivan construction. By \cite[Proposition 17.10]{FHT}, 
we have that $\pi_k(X)\cong \hom(V_k,\R)$, where $V_k$ denotes the subspaces of elements in degree $k$.  
Again, by \cite[Proposition 17.10]{FHT}, we have 
$$\pi_n(F_k)\simeq \left\{\begin{array}{cc}
0 & \text{if } n\neq k\;,
\\
\hom(V_k,\R) & \text{if } n=k\,.
\end{array}\right.
$$
Hence, $F_k\simeq K(\pi_k(X),k)$. Finally, each $X_k$ is $k$-truncated since $\pi_*(X_k)\cong \hom(V_{\leq k},\R)$. 
It follows that $\{X_k\}$ define a Postnikov system for $X$.
\hfill \end{proof}
 
Recall from Section \ref{Sec-diffcoh} that the homotopy type associated to flat $\Lambda V^{\leq k}$-valued 1-forms $\Omega^1_{\rm f\,l}(-;\Lambda V^{\leq k})$ 
  is a presentation for the Sullivan construction. Hence, from Lemma \ref{postsul}, we have a canonical map
  \vspace{-2mm} 
 $$
 \xymatrix{
 \smallint :\Omega_{\rm f\,l}^1(-;\Lambda V^{\leq k})\ar[r] & \delta((X_{\R})_k)
 }
 $$
 
 \vspace{-2mm} 
\noindent which is induced by the unit of the adjunction $\delta\dashv \vert \cdot \vert$ 
(see the discussion around \eqref{realize}).

\medskip
Now we observe that $\R$-localization is also compatible with the Postnikov process. 
We first prove the following lemma.

\begin{lemma}[Compatibility with $\R$-localization]
\label{ratfibers}
Let $X\to Y$ be a Kan fibration, with fiber $F$. Suppose, moreover, that $Y$ is simply connected and
 that $Y$ or $F$ are of rational finite type. Let $L_{\R}$ denote the (derived) functor that localizes at $\R$. 
 Then $L_{\R}X\to L_{\R}Y$ is a Kan fibration with fiber $L_{\R}F$. 
\end{lemma}
\proof
First, $\R$ localization can be computed by applying the functor $A_{PL}:\sset\to {\rm DGCA}$, 
which takes piecewise linear forms on the singular simplicial set associated to $X$, and composes with the (derived)
 Sullivan construction ${\cal K}:{\rm CDGA}^{\rm op}\to \sset$. By \cite[Proposition 15.5]{FHT} 
 (see also \cite[Theorem 2.2]{Hess}), the Sullivan replacement for sequence $A_{PL}(Y)\to A_{PL}(X)\to A_{PL}(F)$ 
 is the cofiber of a relative Sullivan inclusions. By \cite[Proposition 17.10]{FHT}, the Sullivan construction sends 
 this to a fibration between corresponding $\R$-localizations. By \cite[Proposition 15.5]{FHT},
 we see also that the fiber of this fibration is precisely the $\R$-localization of $F$.
\hfill \endproof

\begin{cor}[Postnikov system over $\R$]
\label{ratpost}
Let $X$ be a Kan complex of rational finite type and let 
 \vspace{-2mm} 
$$
\hdots \longrightarrow X_k\longrightarrow X_{k-1}\longrightarrow \hdots \longrightarrow X_0
$$

 \vspace{-2mm} 
\noindent be a Postnikov system for $X$. Let $L_{\R}$ denote the $\R$-localization functor. Then 
 \vspace{-2mm} 
$$
\hdots \longrightarrow L_{\R}X_k\longrightarrow L_{\R}X_{k-1}\longrightarrow \hdots \longrightarrow L_{\R}X_0
$$

 \vspace{-2mm} 
\noindent
is a Postnikov system for $L_{\R}X$. 
\end{cor}
\proof
Clearly $L_{\R}$ preserves truncation degree, since it induces the rationalization map on homotopy groups.  
By Lemma \ref{ratfibers}, it also sends the fibrations $X_k\to X_{k-1}$ to fibrations 
$L_{\R}X_{k}\to L_{\R}X_{k-1}$, with fiber $L_{\R}F_k$. We also have that 
$L_{\R}F_k\simeq K(\pi_k(L_{\R}X),k)$, which follows immediately from the fact that 
$\pi_*(L_{\R}F)\cong \pi_*(F)\otimes \R$ and $F\simeq K(\pi_k(X),k)$. It remains to show that 
$L_{\R}X_k$ approximate the homotopy type of $L_{\R}X$. Since $X_k$ approximate the homotopy 
type of $X$ and $\pi_*(L_{\R}X_k)\cong \pi_*(X_k)\otimes \R$, we have 
 \vspace{-2mm} 
$$
\pi_*\Big(\lim_{k\to \infty}L_{\R}X_k\Big)\cong
 \lim_{k\to \infty}\pi_*(X_k)\otimes \R\cong \pi_*(X)\otimes \R\;.
$$
Hence, the induced map $L_{\R}X\to \lim_{k\to \infty}L_{\R}X_k$ is a weak equivalence. 
\hfill \endproof

We now turn to the differential refinement of Postnikov systems. Although there is a well-defined 
notion of Postnikov tower which is intrinsic to smooth stacks
(see \cite[Sec. 5.5]{Lurie}), this tower does 
not give quite the right information when passing to the differential refinement. We would really like a tower 
which converges to the refinement $\widehat{X}$ and which is compatible with the pullback property of 
$\widehat{X}$. Motivated by this, we introduce the notion of the \textit{differential} Postnikov tower. 
\begin{defn}[Differential Postnikov systems]
Let $X$ be a simply connected space of rational finite type and let $(\Lambda V,d)$ be a Sullivan model 
for $X_{\R}$. Consider the homotopy pullback diagram of smooth stacks
 \vspace{-2mm} 
$$
\xymatrix@R=1.5em{
\widehat{X}\ar[rr]\ar[d]^-{\smallint}&& \Omega^1_{\rm f\,l}(-;\Lambda V)\ar[d]^-{\smallint}
\\
\delta(X)\ar[rr]^-{\delta(L_{\R})} && \delta(X_{\R})
}
$$

 \vspace{-2mm} 
\noindent
where the $\smallint$'s appearing are the respective components of the unit of the adjunction 
$\vert \cdot \vert\dashv \delta$ 
and $L_{\R}$ is the localization at $\R$. A \textit{differential} Postnikov system for $\widehat{X}$ is
 sequence of smooth stacks
\(\label{postmps}
\hdots \to (\widehat{X})_k\longrightarrow (\widehat{X})_{k-1}\longrightarrow 
\hdots \longrightarrow (\widehat{X})_0\;,
\)

 \vspace{-2mm} 
\noindent
such that for each $k$, $(\widehat{X})_k$ fits into a homotopy cartesian square
 \vspace{-2mm} 
\begin{equation}\label{postlevel}
\xymatrix@R=1.5em@C=4em{
(\widehat{X})_k\ar[r]\ar[d]
 & \Omega^1_{\rm f\,l}(-;\Lambda V^{\leq k})\ar[d]
\\
\delta((X)_k)\ar[r] & \delta((X_{\R})_k)
}\;,
\end{equation}

 \vspace{-2mm} 
\noindent
with each vertex representing the corresponding $k$th Postnikov section
 (in spaces, rational spaces, and DGCA's), and the maps \eqref{postmps} 
are universal maps of smooth stacks induced by pullback.
\end{defn}

\begin{prop}[Compatibility of differential refinement with Postnikov construction] 
\label{prop-compat}
Let $X$ be a simply connected space of rational finite type and let $X_{\R}$ denote its localization 
at $\R$. Fix a Sullivan model $(\Lambda V,d)$ of $X_{\R}$ and let 
$$
\xymatrix{
\hdots \ar[r] & (\widehat{X})_k \ar[r]&  (\widehat{X})_{k-1}\ar[r] & \hdots \ar[r] & (\widehat{X})_0
}
$$ 

\vspace{-1mm} 
\noindent be a differential 
Postnikov system for $\widehat{X}$. Then the system satisfies the following properties:
\begin{enumerate}[{\bf (i)}]
\vspace{-1mm} 
\item The tower converges to $\widehat{X}$, i.e.,
$\lim_{k\to \infty}\widehat{X}_k\simeq\widehat{X}$.
\vspace{-2mm}
\item For $\pi_{n+1}(X)$ a torsion group, the map $\widehat{X}_{n+1}\to \widehat{X}_{n}$ 
has fiber $\delta(K(\pi_{n+1}(X),n+1))$ and is classified by the $k$-invariant 
$$
\xymatrix{
\widehat{X}_n\ar[r]^-{\smallint}&  \delta(X_n)\ar[r] & \delta(K(\pi_{n+1}(X),n+2))
}.
$$
\vspace{-6mm}
\item For $\pi_{n+1}(X)$ free of rank $m$, the map $\widehat{X}_{n+1}\to \widehat{X}_{n}$ 
has fiber $K(\pi_{n+1}(X),n+1)$ and fits into a pullback diagram of the form
\(\label{dfcokin}
\xymatrix{
\widehat{X}_{n+1}\ar[rr]\ar[d] && \prod_{i=1}^m\Omega^{n}\ar[d]^-{\prod_{i=1}^m a}
\\
\widehat{X}_n\ar[rr] && \prod_{i=1}^m\mathbf{B}^{n+1}U(1)_{\nabla}
\;.}
\)
Here $a:\Omega^n\to \mathbf{B}^{n+1}U(1)_{\nabla}$ is part of the data of the differential 
refinement, whose curvature gives the exterior derivative. The bottom map in \eqref{dfcokin} 
refines the topological $k$-invariant and the $k$-invariant of DGCAs. 
\end{enumerate}
\end{prop}
\begin{proof}
(i) By Corollary \ref{ratpost}, the localization $L_{\R}$ sends a Postnikov system for $X$ to a Postnikov system for $X_{\R}$. 
Since homotopy limits commute, we can commute the homotopy pullbacks in \eqref{postlevel} with the sequential limit in \eqref{postmps}. Then we have 
$$
\lim_{k\to \infty}(\widehat{X})_k\simeq \widehat{X}\;.
$$
(ii) By the Whitehead theorem, $L_{\R}K(\pi_{n+1}(X),n+1)$ is weakly contractible whenever $\pi_{n+1}(X)$ is torsion. Commuting homotopy fibers and homotopy pullbacks, we see that the map classifying the extension is the homotopy 
pullback of the corresponding classifying maps in the topological, 
rational, and DGCA case (see \eqref{postlevel}). But since $L_{\R}K(\pi_{n+1}(X),n+1)$ is contractible, and 
$\delta$ is homotopy continuous, the homotopy pullback of homotopy fibers is given by
 $$
 \xymatrix@R=1.5em@C=3em{
 \delta(K(\pi_{n+1}(X),n+1))\ar[d]\ar[r] & \ast \ar[d]
 \\
  \delta(K(\pi_{n+1}(X),n+1))\ar[r] & \ast.
  }
 $$ 
Hence, the refinement of the $k$-invariant at this stage collapses to the purely topological 
case, as claimed. 
(iii) Let $\{v_i\}_{i=1}^m$ be a basis for $\pi_{n+1}(X)\otimes \R$. 
Then, in DGCA's, the extension is classified by the pushout diagram 
\begin{equation}\label{pushoutdga}
\xymatrix{
\Lambda V^{\leq n+1} &&\ar[ll] \Lambda \big(w_1,w_2,\hdots, w_m,dw_1, dw_2,\hdots dw_m \big)
\\
\Lambda V^{\leq n}\ar[u] && \ar[ll]_\varphi \R[v_1,v_2,\hdots v_m]\ar[u]_\phi
}
\end{equation}
where the bottom map $\varphi$ dualizes to the classifying map and the right vertical arrow $\phi$
 is defined by sending $v_i\mapsto dw_i$. Taking DGCA homomorphisms to forms, 
 $\hom_{\rm dgca}(-;\Omega^*)$, gives 
 a corresponding pullback diagram 
\(\label{pulldgca}
\xymatrix{
\Omega^*(-;\Lambda V^{\leq n+1})\ar[rr]\ar[d] && \prod_{i=1}^m\Omega^n\ar[d]^-{\prod_{i=1}^md}
\\
\Omega^*(-;\Lambda V^{\leq n})\ar[rr] && \prod_{i=1}^m\Omega^{n+1}_{\rm cl}.
}
\)
The identification of the top right corner follows by observing that a DGA homomorphism 
$$
\Lambda(w_1,\hdots,w_m,dw_1,\hdots,dw_m)\longrightarrow \Omega^*
$$
 is completely determined by where it sends the $w_i$'s. 
We now apply the realization functor $\vert \cdot \vert$ (c.f. \eqref{realize}) to the above diagram and use the 
fact that $\vert \Omega^*_{\rm fl}(-;\Lambda V) \vert\simeq {\cal K}(\Lambda V)$, where ${\cal K}$ denotes 
the Sullivan construction (see section \ref{Sec-DiffCoh}). Observe that the pushout diagram \eqref{pushoutdga} 
is already a derived pullback square in ${\rm DGCA}^{\rm op}$, with respect to the model structure defined in 
\cite[Section 4]{BG}, since three of the objects are Sullivan algebras and the left vertical map is an inclusion of 
relative Sullivan algebras. In \cite[Section 8]{BG} (see \cite[page 9]{Hess} for a review), it was shown that the 
Sullivan construction functor 
$$
\mathcal{K}:{\rm DGCA}^{\rm op}\longrightarrow \sset
$$
 is a right Quillen functor (with respect to the usual Quillen-Kan model structure on $\sset$). Therefore, 
 ${\cal K}(-)\simeq \vert \Omega^*_{\rm fl}(-;(-))\vert$ sends this homotopy pullback to a homotopy
 pullback diagram
\(
\label{diag-hpbs}
\xymatrix@R=1.5em{
(X_{\R})_{n+1}\ar[rr]\ar[d] && \ast \ar[d]
\\
(X_{\R})_n \ar[rr] && \prod_{i=1}^mK(\R,n+1)\;.
}
\)
Since $\pi_{n+1}(X)$ is given to be free of rank $m$, we also have a homotopy fiber sequence 
\(\label{int-hpbs}
\xymatrix@R=1.5em{
(X)_{n+1}\ar[rr]\ar[d] && \ast \ar[d]
\\
(X)_n \ar[rr]^-{k} && \prod_{i=1}^mK(\Z,n+1)
}
\)
which $\R$-localizes to the fiber sequence \eqref{diag-hpbs} above. Finally, the $k$-invariant 
in the differential Postnikov tower is just the homotopy pullback of the corresponding $k$-invariants in 
spaces, $\R$-local spaces and CDGA's. Now comparing the $k$-invariants 
in \eqref{int-hpbs}, \eqref{diag-hpbs} and \eqref{pulldgca}, we find the homotopy 
pullback is given by 
$$
\xymatrix@R=1.5em{ 
\prod_{i=1}^m\mathbf{B}^{n+1}U(1)_{\nabla}\ar[rr]\ar[d] && 
\prod_{i=1}^m\Omega^{n+1}_{{\rm cl}}\ar[d]
\\
\prod_{i=1}^m \delta(K(\Z,n+1))\ar[rr] && \prod_{i=1}^m\delta(K(\R,n+1))\;.
}
$$
Then from the diagram \eqref{pulldgca} and the commutative triangle 
$$
\xymatrix@R=.5em{
\Omega^n\ar[rr]^-{d}\ar[rd]_-{a} && \Omega^{n+1}_{\rm cl}
\\
&\mathbf{B}^{n+1}U(1)_{\nabla}\ar[ru]&
}
$$

\vspace{-3mm}
\noindent we see that $\widehat{X}_{n+1}$ fits into the homotopy pullback diagram \eqref{dfcokin} as claimed.
\hfill \end{proof}

%where $(-)_k$ denotes the $k$-the level in the respective Postnikov towers. 
%\begin{proof}
%\end{proof}
\begin{remark}[Extension of the 4-sphere algebra and quaternionic Hopf fibration]\label{4sphralg}
The only nontrivial extension in the Postnikov approximation to the Sullivan algebra 
${\rm CE}(\mathfrak{s}^4)$ occurs in degree $n=5$, where we get a pushout diagram 
$$
\xymatrix@R=1.5em{
{\rm CE}(\mathfrak{s}^4) &&\ar[ll] \Lambda(g_7,dg_7)
\\
\ar[u] \R[g_4] && \R[g_8]\ar[ll]\ar[u]
}
$$
in which the bottom map sends $g_8\mapsto g_4^2$ and the right map sends $g_8\mapsto dg_7$. 
Rationally, this level corresponds to the quaternionic Hopf fibration generating $\pi_7(S^4)\otimes \R$
(see \cite{FSS-twist}\cite{FSS-level}). See also Remark \ref{Rem-QuatHopfFib}.
\end{remark}

We immediately have the following corollary

\begin{cor}[Refinement vs. Postnikov for the 4-sphere] 
\label{dfposs4}
The $n$th section of the differential Postnikov tower takes the form
$$
\xymatrix@R=1.5em@C=4em{
(\widehat{S}^{\,4})_n\ar[r]\ar[d]
 & \Omega^1_{\rm f\,l}\big(-;(\mathfrak{s}^4)^{\leq n}\big)\ar[d]
\\
\delta((S^4)_n)\ar[r] & \delta((S^4_{\R})_n)
}
$$
\begin{enumerate}[{\bf (i)}]
\vspace{-2mm}
\item 
As $n\to \infty$, 
we have 
$
\lim_{k\to \infty}(\widehat{S}^{\, 4})_k=\widehat{S}^{\, 4}\;.
$
\vspace{-2mm}
\item 
Moreover, for $\pi_{n+1}(X)$ torsion, the $k$-invariant at the $n$th stage of the Postnikov system 
for $S^4$ refines to a $k$-invariants for $\widehat{S}^{\, 4}$ via the canonical map 
\vspace{-2mm}
$$
\xymatrix{
(\widehat{S}^{\, 4})_n\ar[r]^-{\smallint}& \delta((S^4)_n)\ar[r]^-{k}& 
\delta(K(\pi_{n+1}(X),n+2))
},
$$
\vspace{-9mm}
\item 
while for $\pi_{7}(S^4)\cong \Z\times \Z_{12}$ the $k$-invariant takes the form 
$$
\xymatrix@R=1.5em{
(\widehat{S}^{\, 4})_3\ar[rr]\ar[d] & & \Omega^7\ar[d]^-{a}
\\
(\widehat{S}^{\, 4})_2\ar[rr] && K(\Z_{12},8)\times \mathbf{B}^7U(1)_{\nabla}\;,
}
$$

\vspace{-4mm} 
\noindent where the projection of the $k$-invariant to the second factor is the Deligne-Beilinson 
cup product $\widehat{G}_4\cup_{\rm DB}\widehat{G}_4$.
\end{enumerate} 
\end{cor}
 \proof
 This follows immediately from Proposition \ref{prop-compat}, setting $S^4=X$ and $\Lambda V=\mathfrak{s}^4$
  (see Remark \ref{4sphralg}). For part (iii), we use the low degree identifications in the (topological) Postnikov 
  tower in Lemma \ref{prop-compat}. In particular, the square cup appears at the $k$-invariant at the second 
  level of the tower, along with an unidentified torsion $k$-invariant. Since the Deligne-Beilinson cup product 
  gives a cup product structure in differential cohomology and uniquely refines 
(up to homotopy) the wedge product of forms and the cup product in integral cohomology
(see \cite{FSS14c}\cite{FSS-cup}\cite{Urs}), part (iii) follows from parts (ii) and (iii) of Proposition \ref{prop-compat}.
 \hfill \endproof

Corollary \ref{dfposs4} gives a complete characterization of the obstruction theory for 
$S^4$ in the differential setting. The $k$-invariants are either purely topological, in the 
torsion case, or are differential refinements of the topological $k$-invariants in the free case. 
As usual, to consider structures on spacetime $Y$, we pull back these universal 
classes and obstruction and evaluate on $Y$.

\begin{prop}
[Differential refinement of Postnikov tower of the sphere] 
 \label{reftower}
To avoid cumbersome notation, we omit the notation $\delta(-)$ for locally constant stacks. 
We simply denote these stacks by the corresponding space. The full differential refinement 
of the Postnikov tower for $S^4$ takes the following form
$$
\xymatrix@C=3em@R=1.5em{
K(\Z_{15},11)\ar[r] & (\widehat{S}^{\, 4})_7\ar[d]
\\
K(\Z_{24} \!\times\! \Z_3,10)\ar[r] & (\widehat{S}^{\, 4})_6\ar[d]\ar[r] & K(\Z_{15},12)
\\
K(\Z_2 \!\times\! \Z_2,9)\ar[r] & (\widehat{S}^{\, 4})_5\ar[d]\ar[r] &K(\Z_{24}\!\times\! \Z_3,11)
\\
K(\Z_2 \!\times\! \Z_2,8)\ar[r] & (\widehat{S}^{\, 4})_4\ar[d]\ar[r] & K(\Z_2 \!\times\! \Z_2,10)
\\
  K(\Z_{12},7)\!\times\! K(\Z,7)\ar[r] &(\widehat{S}^{\, 4})_3 \ar[d] \ar[r] & K(\Z_2\!\times\! \Z_2,9)
\\
   K(\Z_2,6)\ar[r] & (\widehat{S}^{\, 4})_2 \ar[d]\ar[r]^-{(-, \; \widehat{\iota}^{\,2}_4)}  &   K(\Z_{12},8)\times \mathbf{B}^7U(1)_{\nabla}
\\
  K(\Z_2,5)\ar[r] & (\widehat{S}^{\, 4})_1 \ar[d]\ar[r]^-{\alpha_7I}  &  K(\Z_2,7)
   \\
  & (\widehat{S}^{\, 4})_0=\mathbf{B}^3U(1)_{\nabla}  \ar[r]^-{{\rm Sq}^2\rho_2I} & K(\Z_2,6)
}
$$
where we have identified the first few obstructions. 
\end{prop} 
\proof
We have already identified the $k$-invariant at the second stage in Corollary \ref{dfposs4}. The other $k$-invariants are all torsion, hence these obstructions follows from Proposition \ref{prop-compat}, part (ii).
\hfill \endproof

\begin{remark}[The obstruction in M-theory via higher bundles with connections]
Note that locally the Deligne-Beilinson cup product in M-theory 
$\widehat{G}_4\cup_{\rm DB}\widehat{G}_4$ gives a 7-bundle with connection 
form locally given by  $C_3^{\rm form}\wedge G_4^{\rm form}$ \cite{E8}\cite{FSS1}\cite{FSS14c}\cite{FSS-cup}. 
From the identification of the $k$-invariant at the second stage in Proposition \ref{reftower} (the Deligne-Beilinson square), it follows that to lift past the 2nd stage in the Postnikov tower for $\widehat{S}^{\, 4}$, this connection must be globally defined. Explicitly, in terms of differential cohomology, we have 
$$
a(C_3^{\rm form}\wedge G_4^{\rm form})=\widehat{G}_4\cup_{\rm DB}\widehat{G}_4\;,
$$
where $a:\Omega^7(Y^{11})\to \widehat{H}^8(Y^{11})$ is the canonical map.
\end{remark}

\begin{remark}[The stable case]\label{stableprop-compat}
The above has been  the treatment in the unstable case, and the discussion goes through essentially verbatim in 
the stable setting, with minor modifications. The only nontrivial modification is to replace DGCAs 
with the correct algebraic model for $\R$-local spectra. The $\R$-localization of a spectrum is simply 
given by smashing with the real Eilenberg-MacLane spectrum $H\R$. \footnote{Note that the unit 
map $S\to H\Z$ induces an equivalence 
$S\R\simeq H\R$.} By the work of Shipley \cite{Shi}, these are equivalent to just differentially 
graded vector spaces. Since there are no non-trivial rational obstructions for $S^4$ (stably), 
these effects are not seen and we will not spell out these details. We simply note that the first 
two properties of Proposition \ref{dfposs4} hold equally well in the stable setting. 
This gives rise to the following proposition.
\end{remark} 

\medskip
\begin{prop}[Differential cohomotopy vs. differential cohomology] \label{cohvint-diff}
Let $Y^{11}$ be an 11-dimensional smooth manifold. Let $I:\widehat{H}^*(Y^{11})\to H^*(Y^{11};\Z)$ 
be the canonical map relating differential cohomology and integral cohomology. 
Then a class $\hat{a}\in \widehat{H}^4(Y^{11})$ lifts to a class $\hat{b}\in \widehat{\pi}_s^{\, 4}(Y^{11})$ 
if and only if the following conditions are satisfied.
\begin{enumerate}[{\bf (i)}]
\vspace{-2mm} 
\item $Sq^2I(\hat{a})\equiv 0 \mod 2$, ~~$\mathcal{P}^1_3I(\hat{a})\equiv 0 \mod 3$.
\vspace{-2mm} 
\item There is a lift $\hat{a}^{\prime}:Y^{11}\to (\widehat{S}^{\, 4})_1$ of $\hat{a}$ such 
that $\alpha_7I(\hat{a}^{\prime}) \equiv 0 \mod 2$.
\vspace{-2mm} 
\item There is a further lift $\hat{a}^{\prime\prime}:Y^{11}\to (\widehat{S}^{\, 4})_2$ such 
that $\beta_8I(\hat{a}^{\prime\prime}) \equiv 0 \mod 8$. In particular, upon mod 2 reduction, 
we have $Sq^4I(\hat{a})=I(\hat{a}\cup_{\rm DB}\hat{a})=a^2\equiv 0 \mod 2$.
\vspace{-2mm} 
\item There is a further lift $\hat{a}^{\prime\prime\prime}$ of $\hat{a}^{\prime\prime}$
 such that $P_{11}I(\hat{a}^{\prime\prime\prime})\equiv 0 \mod 2$. In particular, upon mod 2 reduction, 
 we have the tautological relation $Sq^8I(\hat{a})\equiv 0 \mod 2$.
\end{enumerate}
\end{prop}
\proof
We use the stable version of Proposition \ref{prop-compat} (see Remark \ref{stableprop-compat}). By the identification of the $k$-invariants in Lemma \ref{Lem-IntPost}, we see that all the invariants are torsion. Hence, part (ii) of Proposition \ref{prop-compat} implies that all the $k$-invariants are purely topological. Unwinding the statements in the present case, we see that the first obstruction is given by ${\rm Sq}^2I:\mathbf{B}^3U(1)_{\nabla}\to K(\Z_2,6)$ (the map $I$ is precisely the unit of the adjunction $\smallint$ appearing in Proposition \ref{prop-compat}). At higher levels, this follows immediately from the topological counterpart Proposition \ref{cohvint} and part (ii) Proposition \ref{prop-compat}, using the notation $I$ in place of $\smallint$.
\hfill \endproof

The following is essentially the differential refinement of Prop. \ref{cor-C}. 
Since the obstructions in the differential Postnikov tower are completely torsion 
in our case, the proof is verbatim the same as Prop. \ref{cor-C}, with 
$I(\widehat{C}_3)$ replacing $\widetilde{C}_3$. 

\begin{prop}
[Differential cohomotopy vs. differential cohomology for the C-field]  
\label{refcor-C}
Consider the differentially refined M-theory (shifted) C-field $\widehat{G}_4$ as an 
integral cohomology class in degree four. Then if $\widehat{G}_4$ lifts to 
a cohomotopy class ${\cal G}_4\in \widehat{\pi}^4(Y^{11})$ the following obstructions 
necessarily vanish 
\begin{enumerate}[{\bf (i)}] 
\vspace{-2mm}
\item ${\rm Sq}^2I(\widehat{G}_4)=0 \in H^6(Y^{11}; \Z_2)$. 
\vspace{-2mm}
\item $\mathcal{P}^1_3I(\widehat{G}_4)=0 \in H^8(Y^{11}; \Z_3)$.
\vspace{-2mm}
\item ${\rm Sq}^4I(\widehat{G}_4)=I(\widehat{G}_4\cup_{\rm DB}\widehat{G}_4)=0  
\in H^{8}(Y^{11}; \Z_{2})$.
\vspace{-2mm}
\item If $\widehat{G}_4=0$ and $C^{\rm form}_3$ is quantized, with differential 
refinement $\widehat{C}_3$, then we also have 
${\rm Sq}^3{\rm Sq}^1I(\widehat{C}_3)=0\in H^7(Y^{11};\Z_2)$.
\vspace{-2mm}
\item
If $dG^{\rm form}_7=G^{\rm form}_4\wedge G^{\rm form}_4=0$ and $G^{\rm form}_7$ is quantized, 
with differential refinement $\widehat{G}_7$, then we also have the condition 
 ${\rm Sq}^4I(\widehat{G}_7)=0\in H^{11}(Y^{11};\Z_2)$.
\end{enumerate}
\end{prop}

\begin{remark} [Interpretation and congruences]
Among the new conditions, Prop. \ref{refcor-C} reproduces the correct mod 2 congruence condition 
for $G_4^{\rm form}\wedge G_4^{\rm form}$, previously obtained using $E_8$-gauge theory 
in \cite{flux}. As indicated in Remark \ref{congM-pat}, we can also obtain the mod 3 congruence 
by considering the top obstruction on a closed 12-manifold $Z^{12}$. 
\end{remark}

\begin{remark}[Differential cohomotopy first contribution to the C-field] 
The interpretation of the degree 5 class in Remark \ref{Rem-pure5} holds equally in the differential case 
and is closely related to the condition ${\rm Sq}^3{\rm Sq}^1I(\widehat{C}_3)$ in Prop. \ref{refcor-C}. 
Indeed, using the Adem relation ${\rm Sq}^3{\rm Sq}^1={\rm Sq}^2{\rm Sq}^2$, we note that this 
condition comes from restricting the secondary obstruction $\alpha_7$ to the fiber, where it acts on the 
degree 5 class ${\rm Sq}^2I(\widehat{C}_3)$ by ${\rm Sq}^2$. 
\end{remark}

\begin{example}[Differential cohomotopy of flux compactification spaces]
We consider the differential cohomotopy of the spacetime backgrounds computed in 
Examples \ref{Flux-ex}. First observe that, by the general machinery of differential 
refinements of generalized cohomology (see \cite{GS3}\cite{GS7}), 
we have a long exact sequence in stable cohomotopy
\vspace{-2mm} 
$$
\xymatrix{
\hdots \ar[r]& \pi_s^3(X)\ar[r]^-{{\rm deg}}& 
\Omega^3(X)\ar[r] & \widehat{\pi}^{\, 4}_s(X)\ar[r] & \pi_s^4(X)\ar[r] & \hdots 
}.
$$
We will use this exact sequence to compute some examples. 
\begin{enumerate}[{\bf (i)}]
\vspace{-2mm} 
\item \underline{$\widetilde{\rm AdS}_7\times \R P^4$}: 
Here we observe that the cofiber sequence 
$
\R P^3\to \R P^4\overset{q_4}{\to} S^4
$.
This gives rise to an exact sequence in cohomology
$$
\xymatrix{
\Z\ar[r]^-{\times 2} & \Z\cong H^4(S^4;\Z) \ar[r] & H^4(\R P^4;\Z)\ar[r]& 
0}
$$
from which we learn that the pullback of the fundamental class of $S^4$ by $q_4$ is the generator 
of $H^4(\R P^4;\Z)\cong \Z_2$. This gives an isomorphism 
$$
\pi^4(\R P^4)\cong H^4(\R P^4;\Z), ~~~~ q_4\mapsto q_4^*\iota_4\;.
$$
Now from the cofiber sequence above, we also compute $\pi^3(\R P^4)\cong 0$.
We therefore have a short exact sequence
$$
\xymatrix{
0\ar[r]& \Omega^3(\R P^4)\ar[r] & \widehat{\pi}^4(\R P^4)\ar[r] & \pi^4(\R P^4)\cong \Z_2\ar[r] & 0
}.
$$
The Five Lemma produces an isomorphism $\widehat{\pi}^{\, 4}(\R P^4)\cong \widehat{H}^4(\R P^4)$. 
Using the fact that ${\rm AdS}_7$ is topologically trivial, this also implies an isomorphism 
$\widehat{\pi}_s^{\, 4}\big(\widetilde{{\rm AdS_7}}\times \R P^4\big)\cong
 \widehat{H}^4\big(\widetilde{{\rm AdS_7}}\times \R P^4\big)$.

\vspace{-2mm} 
\item \underline{$\widetilde{\rm AdS}_4\times \C P^2$}: Here we recall that, 
in the topological case, $\pi^4(\C P^2)\cong \Z$ with generator $[q_2]$. The `realification' map 
$\Z\cong \pi^4(\C P^2)\to \pi^4(\C P^2)\otimes \R\cong \R$ is the canonical inclusion. 
It is easy to check that pullback by $q_2:\C P^2\to S^4$ induces an isomorphism on $H^4$. 
Hence, $\Omega^4(\C P^2)\to \pi^4(\C P^2)\otimes \R$ maps a closed form $\omega_4$, 
generating $H^4(\C P^2;\R)\cong \R$ to the generator $[q_2]$. Using the Hopf fibration, 
one can show that $\pi^3(\C P^2)\cong 0$. In this case, these considerations lead to a short 
exact sequence 
$$
\xymatrix{
0\ar[r] & \Omega^3(\C P^2)\ar[r] & \widehat{\pi}_s^{\, 4}(\C P^2)\ar[r] & \pi^4(\C P^2)\ar[r] & 0
},
$$
and the Five Lemma produces an isomorphism 
$\widehat{\pi}_s^{\, 4}(\C P^2)\cong \widehat{H}^4(\C P^2)$.
Using the fact that ${\rm AdS}_4$ is topologically trivial, this also implies an isomorphism 
$\widehat{\pi}_s^{\, 4}\big(\widetilde{{\rm AdS_4}}\times \C P^2\big)\cong
 \widehat{H}^4\big(\widetilde{{\rm AdS_4}}\times \C P^2\big)$.

\vspace{-2mm} 
\item \underline{$\widetilde{\rm AdS}_4\times \C P^2\times T^2$}: In this case, $T^2$ does not 
contribute to $\pi^4$ or $\pi^3$ topologically (as in Examples \ref{Flux-ex}). 
Then the same argument as in part {\bf (ii)} above gives 
$$
\widehat{\pi}_s^{\, 4}\big(\widetilde{\rm AdS}_4\times \C P^2\times T^2\big)\cong 
\widehat{H}^4\big(\widetilde{\rm AdS}_4\times \C P^2\times T^2\big)\;.
$$

\item \underline{$\widetilde{\rm AdS}_4 \times \R P^5 \times T^2$}: 
As noted in part {\bf (iv)} of Examples \ref{Flux-ex}, 
$\pi^4(\R P^5)$ is order 4, either $\Z_4$ or $\Z_2 \times \Z_2$, while 
$H^4(\R P^5; \Z)\cong \Z_2$. From \cite{We}, $\pi^3(\R P^5)$ 
is finite. We therefore have a short exact sequence
$$ 
\xymatrix{
0\ar[r] & \Omega^3(\R P^5)\ar[r] & \widehat{\pi}^{\,4}(\R P^5)\ar[r] & \pi^4(\R P^5)\ar[r] & 0
}.
$$
Since $\pi^4(\R P^5)$ is generated by $q_5\eta_4$, with $\eta_4:S^5\to S^4$ the two-fold suspension 
of the Hopf map, the induced map on $H^4$ necessarily vanishes. Hence, in this case, differential cohomotopy 
yields considerably different information than ordinary differential cohomology. 
\end{enumerate}
\end{example}

%%%%%%%%%%%%%%%

  \vspace{.5cm}
\noindent Daniel Grady, {\it Department of Mathematics, Texas Tech University, Lubbock, TX 79409, USA.}

 \medskip
\noindent Hisham Sati, {\it Mathematics, Division of Science, New York University Abu Dhabi, UAE.}

\end{document}